\documentclass[12pt]{article}
\usepackage{amsmath}
\usepackage{graphicx}
\usepackage{enumerate}
\usepackage{natbib}
\usepackage{url} 
\usepackage{preamble}
\usepackage[english]{babel}
\usepackage[utf8]{inputenc}
\usepackage{mathtools}
\usepackage{amsmath,amsthm,amssymb,bbm}
\newtheorem*{theorem*}{Theorem}
\newtheorem*{prop*}{Proposition}
\newtheorem{propS}{Proposition}
\newtheorem{theoremS}{Theorem}

\newcommand{\blind}{1}

\begin{document}

\def\spacingset#1{\renewcommand{\baselinestretch}%
{#1}\small\normalsize} \spacingset{1}


\if1\blind
{
  \title{\bf Efficient Modeling of Quasi-Periodic Data with Seasonal Gaussian Process}
  \author{Ziang Zhang \hspace{.2cm}\\
    Department of Statistical Sciences, University of Toronto\\
    and \\
    Patrick Brown \\
    Department of Statistical Sciences, University of Toronto\\
      Centre for Global Health Research, St Michael's Hospital \\
    and \\
    Jamie Stafford \\
    Department of Statistical Sciences, University of Toronto\\}
  \maketitle
} \fi

\if0\blind
{
  \bigskip
  \bigskip
  \bigskip
  \begin{center}
    {\LARGE\bf Efficient Modeling of Quasi-Periodic Data with Seasonal Gaussian Process}
\end{center}
  \medskip
} \fi

\bigskip
\begin{abstract}
Quasi-periodicity refers to a pattern in a function where it appears periodic but has evolving amplitudes over time. 
This is often the case in practical settings such as the modeling of case counts of infectious disease or the carbon dioxide (CO2) concentration over time.
In this paper, we introduce a class of Gaussian processes, called seasonal Gaussian Processes (sGP), for model-based inference of such quasi-periodic behavior. 
We illustrate that the exact sGP can be efficiently fit within $O(n)$ time using its state space representation for equally spaced locations. 
However, for large datasets with irregular spacing, the exact approach becomes computationally inefficient and unstable. To address this, we develop a continuous finite dimensional approximation for sGP using the seasonal B-spline (sB-spline) basis constructed by damping B-splines with sinusoidal functions. 
We prove that the proposed approximation converges in distribution to the true sGP as the number of basis functions increases, and show its superior approximation quality through numerical studies.
We also provide a unified and interpretable way to define priors for the sGP, based on the notion of predictive standard deviation (PSD).
Finally, we implement the proposed inference method on several real data examples to illustrate its practical usage. 
\end{abstract}

\noindent%
{\it Keywords:}  Hierarchical Model, Bayesian Methods, Statistical Computing, Prior Elicitation
\vfill

\newpage
\spacingset{1.9} 
\section{Introduction}\label{sec:intro}

In astronomical studies the brightness of Sun-like stars tends to exhibit seasonal variation with amplitudes evolving overtime owing to the rotational modulation of magnetically active regions \citep{roberts2013gaussian}.
This kind of quasi-periodic behavior can also be observed in the weekly mortality counts due to influenza, and the carbon dioxide concentration (CO2) readings at the Mauna Loa observatory \citep{rasmussen2003gaussian}.
Therefore, inferring an unknown function $g$ with quasi-periodic behavior is often a crucial task in a variety of contexts.

In this paper, we proposed a class of Gaussian processes ($\GP$), which we call seasonal Gaussian Processes (sGP), to make inference of a quasi-periodic function $g$. 
The sGP has a natural construction through an ordinary stochastic differential equation (SDE) which provides a clear interpretation of the degree of quasi-periodicity in the inferred function. 
Additionally, once written into a state space representation with its first derivative, the sGP has a Markov property that simplifies the computation for inference at equally spaced locations.

However, for large datasets with irregular spacing, the exact approach can become computationally inefficient and numerically unstable.
To address this, we develop a novel finite dimensional continuous approximation to the sGP with the seasonal B-spline basis (sB-spline). Motivated by the covariance structure of the sGP, sB-splines are constructed by damping B-splines with sinusoidal functions. 
The proposed sB-spline approximation is then constructed using the finite element method {(FEM)}\citep{brenner2008mathematical} in a strategy similar to \cite{rw208}, \cite{spde} and \cite{iwpus} for approximating different classes of Gaussian processes.
Improvements in computational efficiency are realized by adopting a least squares approach so the resulting approximation is fully determined by a set of basis weights with a highly sparse precision matrix. This approximation is shown to converge rapidly to the true sGP as the number of basis functions increases.

The remainder of the paper is organized as follows. 
In Section \ref{sec:sGP}, we introduce the sGP model and its SDE characterization. 
We then derive some useful properties of the sGP, in particular a state-space representation which makes the sGP ideal for inference. 
In Section \ref{sec:FEM}, we develop our proposed sB-spline approximation to the sGP using the FEM. 
We also discuss the computational and theoretical properties that justify its practical advantages. 
In Section \ref{sec:prior}, we derive the predictive standard deviation (PSD) of the sGP \citep{iwpus}, and use it to define prior for the standard deviation parameter of the sGP in a unified and interpretable way. 
In Section \ref{sec:examples}, we demonstrate the effectiveness and versatility of our proposed methodology for inferring quasi-periodic functions. 
We first apply the approach to three simple examples as a proof of concept, and then showcase its advantage by re-analyzing the atmospheric Carbon Dioxide (CO2) concentrations data collected in Hawaii, which was originally analyzed in \cite{rust1979inferences}. 
Our inference uses the posterior approximation algorithm in \cite{noeps, aghqtheory}.
For a more detailed description of how to implement this algorithm with $\GP$, one can refer to Section 3 in \cite{iwpus}.

\section{Inference with the Seasonal Gaussian Process}\label{sec:sGP}

We consider the following hierarchical model:
\begin{equation}\label{equ:smoothModel}
    \begin{aligned}
    Y_i|\boldsymbol{\eta} &\overset{ind}{\sim} \pi(Y_i|\boldsymbol{\eta}, \kappa),  \\
    \eta_i &= \boldsymbol{v}_i ^T \boldsymbol{\beta} + \sum_{l=1}^{L}g_{l}(x_{li}),   \\
    g_l:\ &\Omega_l \rightarrow \mathbb{R}, \  g_l \overset{ind}{\sim} \GP(\C_l), \forall l \in [L].
    \end{aligned}
\end{equation}
Here $\pi(Y_i|\boldsymbol{\eta}, \kappa)$ is a twice continuously-differentiable density with linear predictors $\boldsymbol{\eta} = [\eta_1, ..., \eta_n]^T$, covariates $\boldsymbol{v}_i, \ x_{li}$ and hyperparameter $\kappa$. 
Each unknown function $g_l$ is assigned an independent Gaussian Process ($\GP$) model with zero mean and covariance function $C_l$. 
The model \ref{equ:smoothModel} is termed the Extended Latent Gaussian Model in \cite{noeps}, and an example is the partial likelihood of Cox proportional hazards model where each $Y_i$ depends on multiple elements of $\boldsymbol\eta$ \citep{coxphus}.
For the purpose of exposition, we consider the inference of a single unknown function $g$ with the sGP in the remainder of this section as well as \cref{sec:approximation} and \cref{sec:prior}.


\subsection{Seasonal Gaussian Process: A Differential Equation Characterization}\label{subsec:sGP}

The seasonal Gaussian process (sGP) is a zero-mean $\GP$ defined for functions that exhibit quasi-periodic behavior.
It is characterized by two parameters, the frequency parameter $\alpha \in \R^{+}$ and the standard deviation (SD) parameter $\sigma \in \R^{+}$. 
If the function $g$ is assigned the sGP model, which we denote as $g \sim \text{sGP}(\alpha, \sigma)$, then $g$ satisfies the following stochastic differential equation (SDE):
\begin{equation}\label{equ:sGPdef}
    \begin{aligned}
    Lg(x) = \bigg[\frac{d^2}{dx^2}+\alpha^2\bigg] g(x) = \sigma \xi(x),
    \end{aligned}
\end{equation}
where $\xi(x)$ is the Gaussian white noise process. 
Note that the null space of the SDE in \cref{equ:sGPdef} is: 
\begin{equation}\label{equ:NullSpace}
    \begin{aligned}
    \text{Null}\{L\} = \text{span}\{\cos(\alpha x), \sin(\alpha x)\},
    \end{aligned}
\end{equation}
which is the space of sinusoidal functions with period $2\pi/\alpha$. 
For identifiability, we define the sGP $g$ in \cref{equ:sGPdef} to satisfy the following initial conditions
\begin{equation}\label{equ:boundaryCondition}
    \begin{aligned}
    g(0) = g'(0) = 0,
    \end{aligned}
\end{equation}
without the loss of generality.
More general initial conditions can be accommodated by adding two global trigonometric functions $\{v_1\cos(\alpha x), v_2\sin(\alpha x)\}$ to $g$, where initial conditions become $g(0) = v_1$ and $g'(0) = v_2 \alpha$.

The sGP model has a direct interpretation in its parameters: $\alpha$ specifies the frequency and the periodicity of the function $g$, and $\sigma$ specifies the degree that $g$ could deviate from the sinusoidal space in \cref{equ:NullSpace}.  
When $\sigma$ is larger, the sample path of sGP can deviate more from the sinusoidal pattern, which allows for more quasi-periodic behavior in the inferred $g$. 
Therefore, the quasi-periodic behavior of $g$ such as in the sunspot variation could be accommodated and quantified by the sGP model.


\subsection{Properties of the Seasonal Gaussian Process}\label{subsec:sGP_Prop} 

\noindent The specific covariance function of the sGP model in \cref{equ:sGPdef} is given in \cref{theorem:sGP_Property}.
\begin{proposition}[Covariance Function of the Seasonal Gaussian Process]\label{theorem:sGP_Property} Let $g \sim \text{sGP}(\alpha, \sigma)$. 
Then $g$ has a covariance function:
\begin{equation}\label{equ:sGP-cov}
\begin{aligned}
        C(x_1,x_2) &= \bigg(\frac{\sigma}{\alpha}\bigg)^2 \bigg[\frac{x_1}{2}\cos(\alpha(x_2-x_1)) - \frac{\cos(\alpha x_2)\sin(\alpha x_1)}{2\alpha} \bigg] \\
        &= \bigg(\frac{\sigma}{\alpha}\bigg)^2 \bigg[ \frac{\cos(\alpha x_2)x_1}{2}\cos(\alpha x_1) + \left(\frac{\sin(\alpha x_2)x_1}{2} - \frac{\cos(\alpha x_2)}{2\alpha}\right)\sin(\alpha x_1) \bigg],
\end{aligned}
\end{equation}
for any $x_1,x_2 \in \mathbb{R}^+$ such that $x_1\leq x_2$.
\end{proposition} 

\noindent
The detailed derivation of \cref{theorem:sGP_Property} is provided in Appendix A. 
Since the covariance function in \cref{equ:sGP-cov} has the form of damped sinusoidal waves, it is clear that the sGP should be viewed as a model for quasi-periodic functions. 
As an illustration, some simulated sample paths from the sGP models with different $\alpha$ are shown in \cref{fig:samples}.
\begin{figure}[!h]
    \centering
                 \subfigure[$\alpha = \pi$: Covariance]{
      \includegraphics[width=0.45\textwidth]{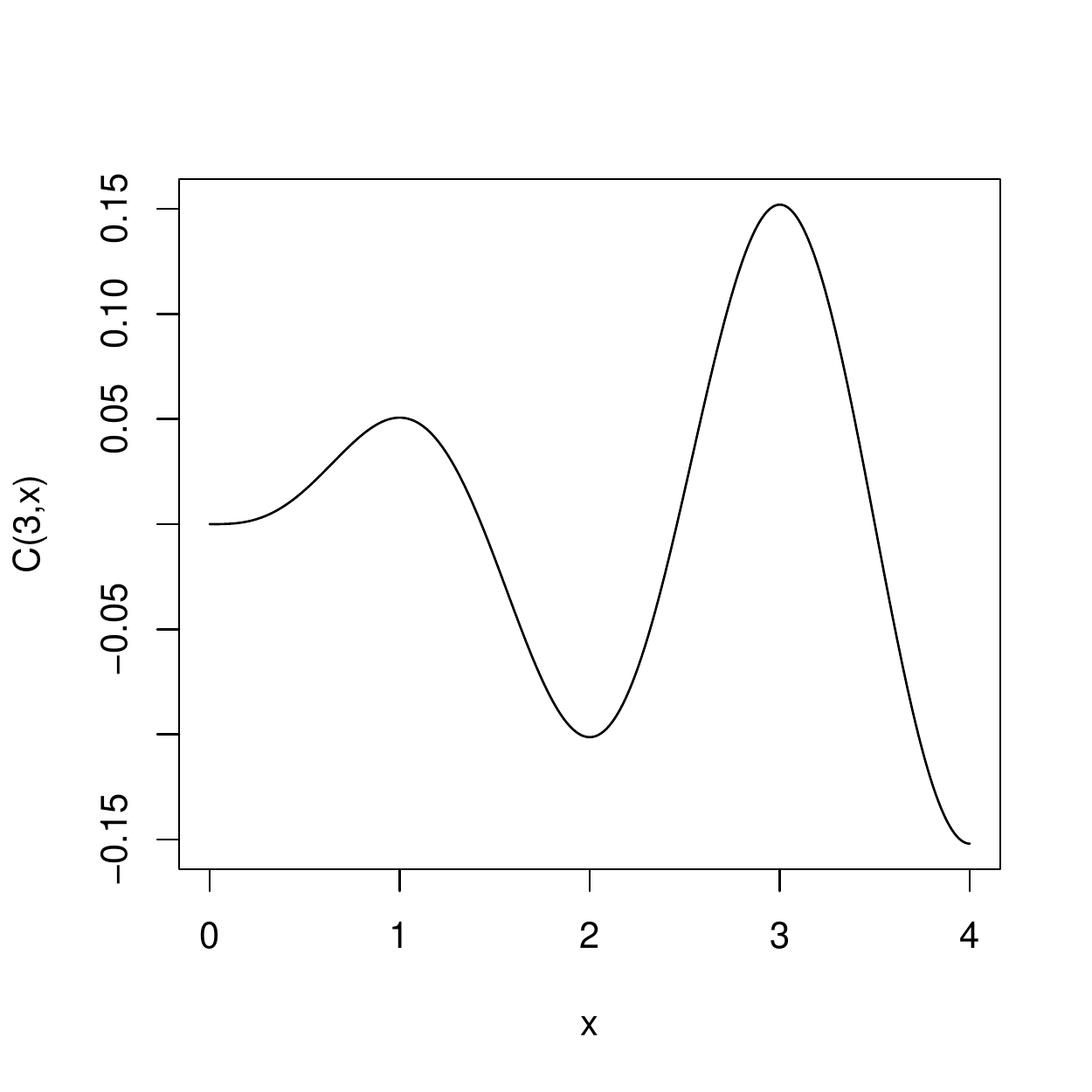}
    }
             \subfigure[$\alpha = 2\pi$: Covariance]{
      \includegraphics[width=0.45\textwidth]{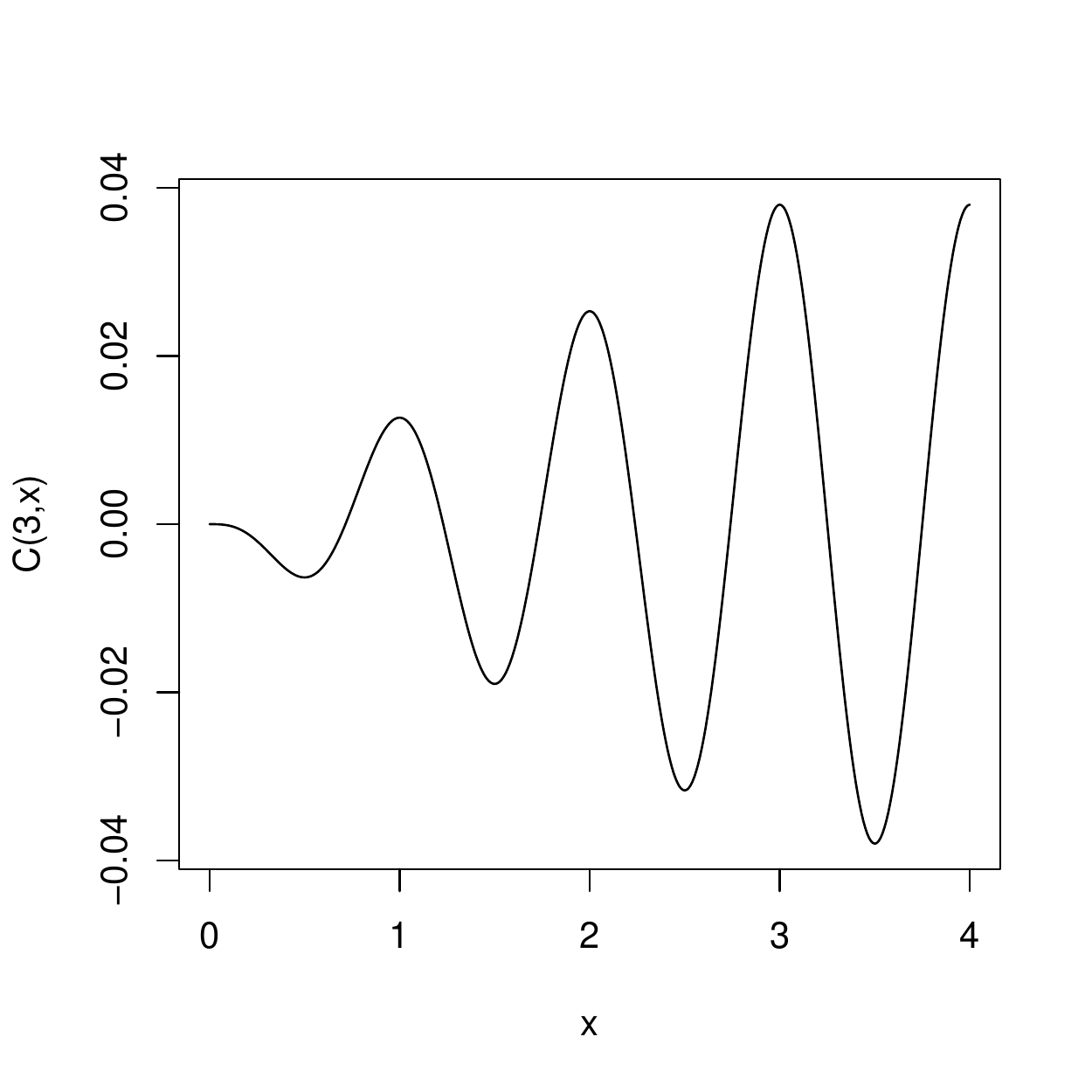}
    }
             \subfigure[$\alpha = \pi$: Samples]{
      \includegraphics[width=0.45\textwidth]{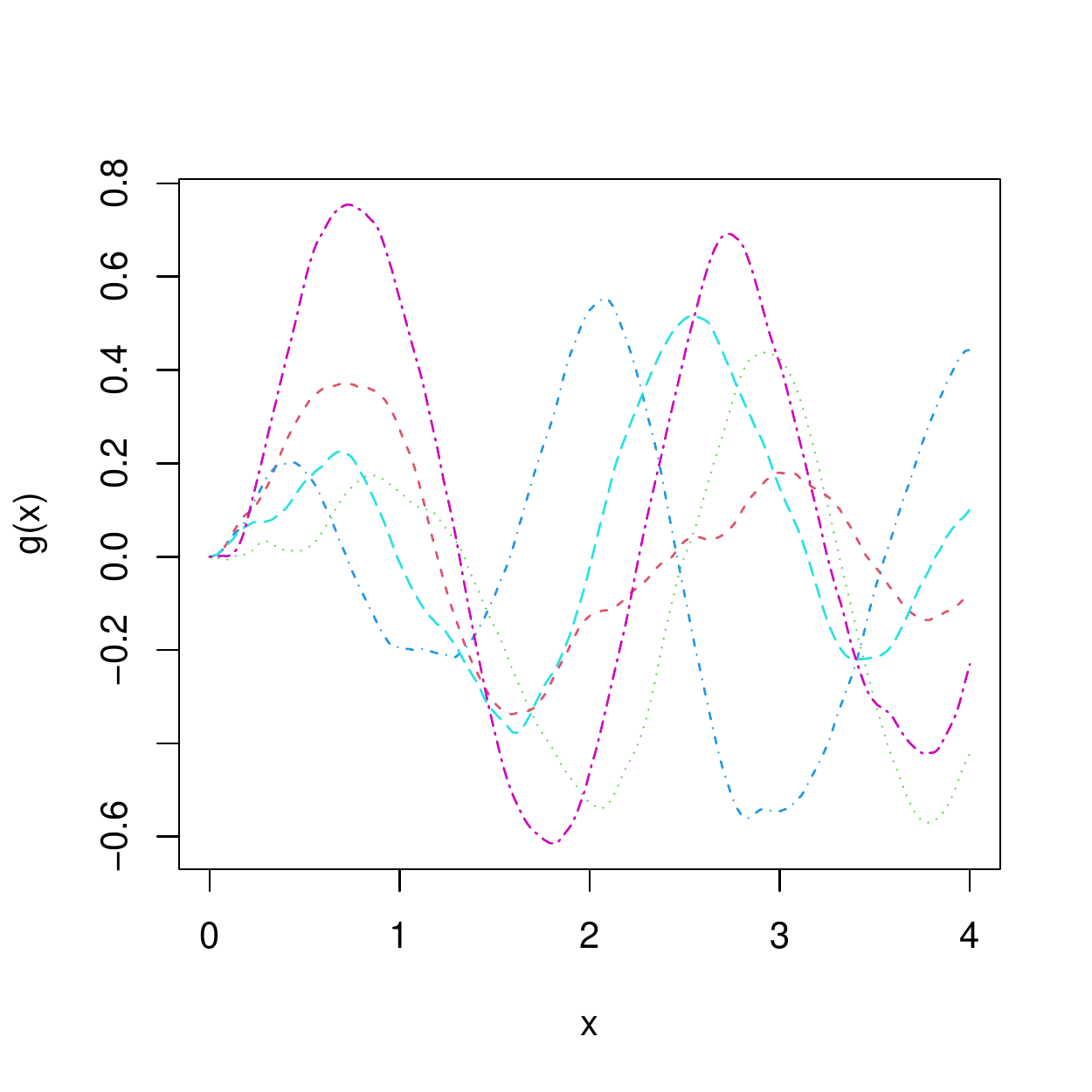}
    }
             \subfigure[$\alpha = 2\pi$: Samples]{
      \includegraphics[width=0.45\textwidth]{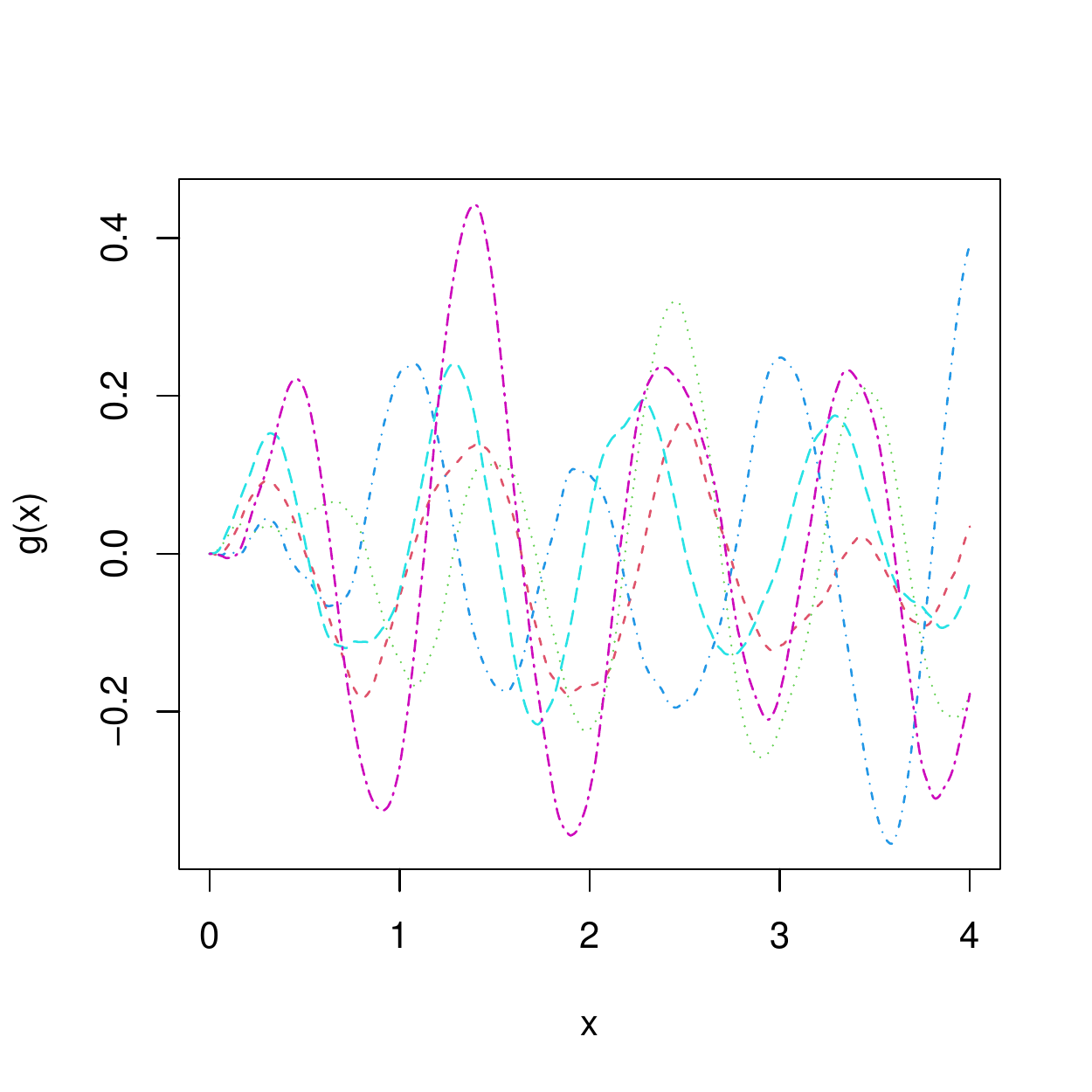}
    }
     \caption{(a,b) display the covariance functions of two sGPs at $3$, and (c,d) display five sample paths from the two sGPs.
     The frequency parameter $\alpha$ equals to $\pi$ in (a,c) and $2\pi$ in (b,d), and the SD parameter $\sigma = 1$ in both sGPs.
     Both the covariance functions and the sample paths exhibit quasi-periodic behavior with amplitudes varying overtime.
     }
    \label{fig:samples}
\end{figure}

Another property that makes sGP convenient to use in practice is that when augmented with its first derivative, the resulting augmented model has a state space representation that satisfies the Markov property. This property is summarized in \cref{thrm:markov}.

\begin{theorem}[State Space Representation of the sGP]\label{thrm:markov} 
Consider $g \sim \text{sGP}(\alpha, \sigma)$, and let $\boldsymbol{s} = \{s_1,..., s_n\} \subset \R^{+}$ denotes a set of $n$ sorted locations and spacing $d_1 = s_1$ and $d_i = s_i - s_{i-1}$ for $i \in \{2,..,n\}$. Then the augmented vector $\boldsymbol{g}_{aug}(s_i) = [g(s_i), g'(s_i)]^T$ can be written as a Markov model:
\begin{equation}\label{equ:state-space}
\boldsymbol{{g}}_{aug}(s_{i+1}) = \boldsymbol{R}_{i+1} \boldsymbol{{g}}_{aug}(s_{i}) + \boldsymbol{\epsilon}_{i+1},
\end{equation}
where $\boldsymbol{\epsilon}_i \overset{ind}{\sim}N(0,\boldsymbol{\Sigma}_i)$. The $2 \times 2$ matrices $\boldsymbol{R}_{i}$ and $\boldsymbol{\Sigma}_{i} = \boldsymbol{Q}_{i}^{-1}$ are respectively defined as:
\begin{equation}
    \begin{aligned}
    \boldsymbol{R}_{i} = \begin{bmatrix} \cos(\alpha d_i) & \frac{1}{\alpha}\sin(\alpha d_i) \\
    -\alpha\sin(\alpha d_i) & \cos(\alpha d_i)  \\
    \end{bmatrix}, \  
    \boldsymbol{\Sigma}_i = \sigma^2 \begin{bmatrix} \frac{1}{\alpha^2}\bigg(\frac{d_i}{2} - \frac{\sin(2\alpha d_i)}{4\alpha}\bigg) & \frac{\sin^2(\alpha d_i)}{2\alpha^2} \\
    \frac{\sin^2(\alpha d_i)}{2\alpha^2} & \frac{2\alpha d_i + \sin(2\alpha d_i)}{4\alpha}  \\
    \end{bmatrix}.
    \end{aligned}
\end{equation}
\end{theorem}
\noindent A proof of theorem \ref{thrm:markov} can be found in Appendix B.
An implication of \cref{thrm:markov} is that the discretely observed sGP model can be fitted sequentially using a filtering method by augmenting the state space with first derivatives. 
This is particularly useful when the response variable $y$ is Gaussian, as demonstrated by \cite{ansley1990filtering} and \cite{ansley1993nonparametric}.
Another implication of \cref{thrm:markov} is that unlike $[g(s_1), ..., g(s_n)]^T \in \mathbb{R}^n$ which has a dense precision matrix, the augmented vector $[\boldsymbol{g}_{aug}(s_1)^T,..,\boldsymbol{g}_{aug}(s_n)^T]^T \in \mathbb{R}^{2n}$ will be a GMRF \citep{gmrfbook} with a larger but much sparser precision matrix $\boldsymbol{Q}_{aug}$ being:
\begin{equation}\label{equ:aug-precision}
\begin{aligned}
\boldsymbol{Q}_{aug} = \left[\begin{smallmatrix} \boldsymbol{Q}_1 + \boldsymbol{A}_1 & \boldsymbol{H}_1 &   &   &   &   \\
                                      \boldsymbol{H}_1^{T} & \boldsymbol{Q}_1 + \boldsymbol{A}_2 & \boldsymbol{H}_2 &   &   &  \\
                                        & \boldsymbol{H}_2^{T} & \boldsymbol{Q}_2 +  \boldsymbol{A}_3 & \boldsymbol{H}_3 &   &  \\
                                        &   & \boldsymbol{H}_3^{T} & \boldsymbol{Q}_3 + \boldsymbol{A}_4 & \boldsymbol{H}_4 &   \\
                                        &   &     &             &  \ddots &   & & \\
                                        &   &     &             &    & \boldsymbol{Q}_{n-2} + \boldsymbol{A}_{n-1} & \boldsymbol{H}_{n-1}\\
                                        &   &     &             & \hdots & \boldsymbol{H}_{n-1}^{T}  & \boldsymbol{Q}_{n-1} \\
                                        \end{smallmatrix}\right],
                                        \end{aligned}
\end{equation}
where $\boldsymbol{Q}_i = \boldsymbol{\Sigma}_i^{-1}$, $\boldsymbol{A}_i = \boldsymbol{R}_i^T\boldsymbol{Q}_i\boldsymbol{R}_i$ and $\boldsymbol{H}_i = - \boldsymbol{R}_i^T\boldsymbol{Q}_i$. 

When all the locations are equally spaced, the Markov model (\ref{thrm:markov}) is reduced to a first-order vector auto-regression (VAR) model which requires only one coefficient matrix $\boldsymbol{R}$ and one error covariance $\boldsymbol{\Sigma}$ to be computed in order to obtain the entire precision matrix $\boldsymbol{Q}_{aug}$.
Because of the sparse structure of $\boldsymbol{Q}_{aug}$, the entire augmented vector $[\boldsymbol{g}_{aug}(s_1)^T,..,\boldsymbol{g}_{aug}(s_n)^T]^T$ can then be efficiently incorporated into the modern approximate Bayesian inference methods of \cite{noeps,inla}.
For large and irregularly spaced datasets, however, the approach using the state-space representation becomes computationally challenging, in particular when the correlation is strong.  This is because of the need to compute and invert each matrix $\boldsymbol{\Sigma}_i$ in \cref{equ:aug-precision}.
It is therefore often useful to consider a continuous finite-dimensional approximation to the entire process, which we introduce in the following section.

\section{Finite Dimensional Approximation to SGP}\label{sec:approximation}


In this section, we derive the proposed seasonal B-spline (sB-spline) approximation as a Least Square approximation through the Finite Element Method (FEM). 
Our results demonstrate that this approximation has both desirable computational and theoretical properties.
In \cref{sec:FEM}, we provide the context and a general form of the FEM approximation. 
For the purpose of computation efficiency, we then adopt a least square approach to construct the FEM approximation and choose the cubic B-spline functions as the basis in \cref{sec:Bspline}.
In \cref{sec:sBspline}, we improve the FEM approximation by damping the basis functions with sinusoidal functions, which finally yields our proposed sB-spline approximation.

\subsection{Finite Element Method}\label{sec:FEM}

Given $\Omega = [a,b]$ for some $a,b \in \mathbb{R}^+$, the idea of FEM is to construct an approximation $\tilde{g}_k$ that has the form of:
\begin{equation}\label{equ:FEMapprox}
    \begin{aligned}
    \tilde{g}_k(x) = \sum_{i=1}^{k} w_i \varphi_i(x),
    \end{aligned}
\end{equation}
where $\boldsymbol{w} = (w_1, ..., w_k)^T \in \R^{k}$ is a set of random weights and $\mathbb{B}_k = \{\varphi_i(x), i \in [k]\}$ is a set of $k$ predetermined basis functions.
Given a set of test functions $\mathbb{T}_k = \{\phi_i(x); i \in [k]\}$, the distribution of the unknown weight vector $\boldsymbol{w}$ is defined such that 
\begin{equation}\label{weakSol}
\begin{aligned}
\langle L\tilde{g}_k(x) , \phi_i(x)\rangle &\overset{d}= \sigma \langle Lg(x) , \phi_i(x)\rangle \\ 
&= \sigma \langle \xi(x) , \phi_i(x)\rangle,
\end{aligned}
\end{equation}
for any test function $\phi_i \in \mathbb{T}_k$, where $\langle f_1(x) , f_2(x)\rangle := \int_\Omega f_1(x)f_2(x) dx$ denotes inner-product between any $f_1,f_2 \in L^2$.
The requirement \cref{weakSol} can be vectorized as
\begin{equation}\label{vectorizedWeakSol}
\begin{aligned}
\textbf{B}\boldsymbol{w} \overset{d}{=} N(\boldsymbol{0}, \sigma^2\textbf{T}),
\end{aligned}
\end{equation}
where $\textbf{B}$ and $\textbf{T}$ are $k \times k$ matrices with entries:
\begin{equation}\label{vectorizedRHS}
    \begin{aligned}
        \textbf{B}_{ij} = \langle \phi_i(x), L\varphi_j(x)\rangle, \quad
        \textbf{T}_{ij} =  \langle \phi_i(x) , \phi_j(x)\rangle.
    \end{aligned}
\end{equation}
When both $\textbf{B}$ and $\textbf{T}$ are non-singular, the requirement in \cref{weakSol} implies $\boldsymbol{w} \sim N(\boldsymbol{0}, \sigma^2 \boldsymbol{\Sigma}_{\boldsymbol{w}})$ where the precision matrix $\boldsymbol{\Sigma}^{-1}_{\boldsymbol{w}} = \textbf{B}^{T} \textbf{T}^{-1} \textbf{B}$.

\subsection{Least Square Approximation with Cubic B-Spline Basis}\label{sec:Bspline}

Given a set of locations $\boldsymbol{s}$, the FEM approximation at these locations can be written as $\tilde{\boldsymbol{g}}_k(\boldsymbol{s}) = \boldsymbol{\Phi} \boldsymbol{w}$ where $\boldsymbol{\Phi}$ is $n \times k$ design matrix with $ij$ element $\boldsymbol{\Phi}_{ij} = \varphi_j(s_i)$.
The computational bottleneck of the above FEM approximation therefore depends on both $\boldsymbol{\Sigma}^{-1}_{\boldsymbol{w}}$ and $\boldsymbol{\Phi}$. We aim to find a FEM approximation such that both $\boldsymbol{\Sigma}^{-1}_{\boldsymbol{w}}$ and $\boldsymbol{\Phi}$ can be computed efficiently and have sparse structures.

To simplify the computation of the precision matrix $\boldsymbol{\Sigma}^{-1}_{\boldsymbol{w}}$, we consider a least square approach to construct the FEM approximation, which corresponds to setting the test function $\phi_i(x) = L \varphi_i(x)$ for each $i \in [k]$. 
This strategy reduces the precision matrix to $\boldsymbol{\Sigma}^{-1}_{\boldsymbol{w}} = \textbf{T}^{T} \textbf{T}^{-1} \textbf{T} = \textbf{T}$, and hence avoids the need to compute another matrix $\textbf{B}$ or to invert the matrix $\textbf{T}$ \citep{iwpus}.

It remains to choose a set of basis functions $\mathbb{B}_k$ such that the corresponding matrices $\textbf{T}$ and $\boldsymbol{\Phi}$ are sparse. Given any basis $\mathbb{B}_k$, it can be shown that the precision matrix $\textbf{T}$ can be computed as:
\begin{equation}\label{equ:least-square-precision}
    \begin{aligned}
    \textbf{T} = a^4\textbf{G} + \textbf{C} + a^2\textbf{M},
    \end{aligned}
\end{equation}
with $\textbf{G}_{ij} =  \langle \varphi_i , \varphi_j \rangle$, $\textbf{C}_{ij} = \langle \frac{d^2\varphi_i}{dx^2} , \frac{d^2\varphi_j}{dx^2} \rangle$ and $\textbf{M}_{ij} = \langle \varphi_i, \frac{d^2\varphi_j}{dx^2} \rangle + \langle \frac{d^2\varphi_i}{dx^2}, \varphi_j \rangle$. 
The detailed derivation is provided in Appendix C.
Given \cref{equ:least-square-precision}, it is natural to consider using the cubic B-spline basis as $\mathbb{B}_k$ in the approximation, which ensures the sparsity of $\textbf{T}$ with $O(8k)$ nonzero elements, and the sparsity of $\boldsymbol{\Phi}$ with $O(4k)$ nonzero elements \citep{Bsplines}.
The cubic B-spline basis yields a computationally efficient approximation to the true sGP. Furthermore, the following \cref{theorem:B-spline-convergence} guarantees the convergence of the approximation to the true sGP process as the number of basis $k$ increases.
\begin{theorem}[Convergence of B-spline Approximation]\label{theorem:B-spline-convergence} Let $\Omega = [a,b]$ where $a,b \in \mathbb{R}^+$ and let $g \sim \text{sGP}(\alpha, \sigma)$. Assume $\mathbb{B}_k$ is a set of $k$ cubic B-splines constructed with equally spaced knots over $\Omega$, and $\tilde{g}_k$ denotes the corresponding FEM approximation defined as in \cref{equ:FEMapprox}, then:
$$\lim_{k\to \infty} \C_{k}(x_1,x_2) = \C(x_1,x_2),$$
for any $x_1,x_2 \in \Omega$, where $\C(x_1,x_2) = \Cov[g(x_1),g(x_2)]$, $\C_{k}(x_1,x_2) = \Cov[\tilde{g}_k(x_1),\tilde{g}_k(x_2)]$.
\end{theorem} 
\noindent The proof of this theorem is provided in Appendix D.

\subsection{Improving Approximation Accuracy with sB-Spline}\label{sec:sBspline}


Although \cref{theorem:B-spline-convergence} guarantees the limiting property of the cubic B-spline approximation, the error from the cubic B-spline approximation could be non-negligible when only a limited number of basis is used, which is a typical situation in real applications.
To see that, note the covariance function (\ref{equ:sGP-cov}) at a fixed location $x_1$ can be written as:
\begin{equation}
    \begin{aligned}
        C_{x_1}(x) = C(x_1,x) = l_{x_1}(x) \cos(\alpha x) + \tilde{l}_{x_1}(x) \sin(\alpha x), \\
    \end{aligned}
\end{equation}
where $l_{x_1}(x)$ and $\tilde{l}_{x_1}(x)$ are piecewise linear functions with a single knot at $x_1$. In order to accurately approximate the distribution of the original sGP, the basis $\mathbb{B}_k$ needs to be flexible enough to approximate well the covariance function $C_{x_1}(x)$ for any $x_1 \in \Omega$. 
The cubic B-spline basis can well approximate the covariance function when $\alpha$ is small relative to the width of $\Omega$. 
However, when $\alpha$ is large relative to the width of $\Omega$, the covariance function will be too oscillating to be approximated well by the cubic B-spline basis. See \cref{fig:fastConverg} for example.

\begin{figure}[!h]
    \centering
             \subfigure[$\alpha = 2\pi/5$]{
      \includegraphics[width=0.45\textwidth]{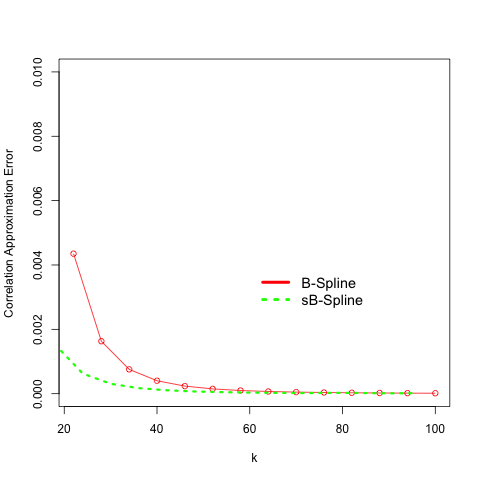}
    }
             \subfigure[$\alpha = 2\pi$]{
      \includegraphics[width=0.45\textwidth]{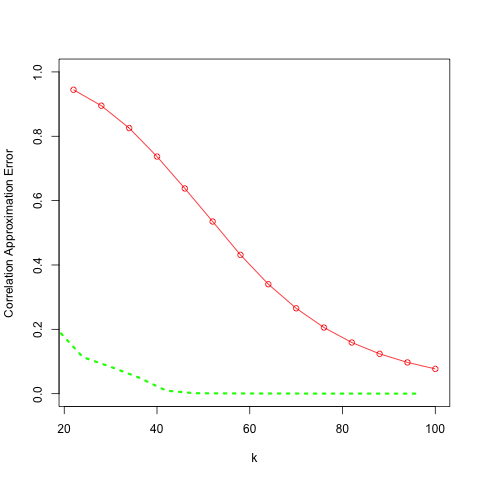}
    }
     \caption{Assessing the approximation accuracy with cubic B-spline (red) and sB-spline (green) basis. 
     In order to make the y-axis comparable, we focus on the correlation function $\rho_{5}(x)$ instead of the covariance function.
     The y-axis shows the maximum correlation approximation error $|\rho_5(x) - \tilde{\rho}_5(x)|$ for $x \in [1,9]$, where $\rho_5(x) = \text{Corr}[g(5),g(x)]$ and $\tilde{\rho}_5(x) = \text{Corr}[\tilde{g}_k(5),\tilde{g}_k(x)]$ are respectively the correlation functions of the true sGP and of the approximation. The interval $\Omega = [0,10]$, $\sigma = 1$ and $\alpha = 2\pi$ or $2\pi/5$.}
    \label{fig:fastConverg}
\end{figure}

Given $C_{x_1}(x)$ is itself a piece-wise polynomial damped by sinusoidal functions, we increase the accuracy of our approximation by augmenting the B-spline basis with B-splines that are also damped by sinusoidal functions.
We call this new basis the seasonal B-spline ($\textit{sB-spline}$) basis. 
The detailed construction of the sB-spline is given as the following. We first construct $r \in \mathbb{Z}^+$ cubic B-spline basis $\{b_i(x)\}_{i=1}^r$ over $\Omega$. Then, we augment the cubic B-spline basis with its damped versions and get
\begin{equation}\label{equ:sBspline}
    \begin{aligned}
     \mathbb{B}_k = \{b_i(x)\}_{i=1}^{r} \cup \{b_i(x)\cos(\alpha x)\}_{i=1}^{r} \cup \{b_i(x)\sin(\alpha x)\}_{i=1}^{r},
    \end{aligned}
\end{equation}
where $k = 3r$. Because the sB-spline basis is directly obtained by damping B-spline it inherits the compact support of the B-spline basis \citep{Bsplines}. Therefore each of $\textbf{G},\textbf{C}$ and $\textbf{M}$ is still a banded sparse matrix, and so will be the precision matrix $\Sigma_{\boldsymbol{w}}^{-1}$ and the design matrix $\boldsymbol{\Phi}$.

While preserving the original computational efficiency from cubic B-spline basis, the sB-spline approximation has significantly better approximation accuracy, particular when $k$ is small and $\alpha$ is large. This can be seen from \cref{fig:fastConverg}; When $\alpha = 2\pi$, the sB-spline approximation reduces the correlation error to less than $0.2$ with $k = 20$, but the original cubic B-spline requires $k$ more than $80$ to achieve the same accuracy. The sB-spline approximation also inherits the convergence result in \cref{theorem:B-spline-convergence}, as summarized in the corollary below:
\begin{corollary}[Convergence of sB-spline Approximation]\label{cor:sB-spline-convergence} 
Let $\Omega = [a,b]$ where $a,b \in \mathbb{R}^+$ and let $g \sim \text{sGP}$ with parameters $\alpha, \sigma >0$.
Assume $\mathbb{B}_k$ is a set of $k$ sB-splines constructed with equally spaced knots over $\Omega$ as defined in \cref{equ:sBspline}, and $\tilde{g}_k$ denotes the corresponding FEM approximation defined as in \cref{equ:FEMapprox}, then:
$$\lim_{k\to \infty} \C_{k}(x_1,x_2) = \C(x_1,x_2),$$
for any $x_1,x_2 \in \Omega$, where $\C(x_1,x_2) = \Cov[g(x_1),g(x_2)]$, $\C_{k}(x_1,x_2) = \Cov[\tilde{g}_k(x_1),\tilde{g}_k(x_2)]$.
\end{corollary}
\begin{proof}
    This corollary directly follows from \cref{theorem:B-spline-convergence} with the observation that \cref{equ:sBspline} contains the original cubic B-spline basis.
\end{proof}

\section{Prior for the sGP}\label{sec:prior}

To perform Bayesian inference with $g\sim \text{sGP}(\alpha,\sigma)$, it is crucial to define a prior on the standard deviation (SD) parameter $\sigma$ in an interpretable manner. However, the effect and interpretation of $\sigma$ on the fitted function vary significantly with the choice of $\alpha$, even though it can always be viewed as the deviation size from the corresponding sinusoidal space $\text{Null}\{L\}$.
Therefore, instead of assigning prior on the original SD $\sigma$, we propose to assign the prior on the predictive standard deviation (PSD) \citep{iwpus}, which is a one-to-one transformation of $\sigma$ with a consistent interpretation across different values of $\alpha$.

Given a once-differentiable function $g$ and a prediction unit $h>0$, the PSD is defined as:
\begin{equation}\label{equ:PSD-general}
    \text{SD}\bigg[g(x+h) \bigg| g(x), g^{(1)}(x)\bigg],
\end{equation}
which quantifies the uncertainty in predicting the function $g$ at $h$ units ahead, using the current value and the value of its derivative. 
In practice, the unit $h$ can be chosen based on the practical requirements of the application.
For the sGP, the PSD has a simple formula, as given in the following corollary:
\begin{corollary}[Predictive Standard Deviation of the sGP]\label{cor:PSD} 
Let $g \sim \text{sGP}$ with parameters $\alpha, \sigma >0$, then:
\begin{equation}\label{equ:PSD-sGP}
    \begin{aligned}
    \sigma(h) &= \text{SD}\bigg[g(x+h) \bigg| g(x), g^{(1)}(x)\bigg]
    = \frac{\sigma}{\alpha} \sqrt{\bigg(\frac{h}{2} - \frac{\sin(2\alpha h)}{4\alpha}\bigg)},
    \end{aligned}
\end{equation}
for any $x,h \in \mathbb{R}^+$.
\end{corollary}
\begin{proof}
This corollary is a direct consequence of the state-space representation in \cref{thrm:markov}.
\end{proof}

It is worth noting that the PSD $\sigma(h)$ of the sGP does not depend on the current location $x$, as shown in \cref{equ:PSD-sGP}. 
This feature distinguishes our approach from that of \cite{scalingigmrf}, which uses the marginal standard deviation $\text{SD}[g(x)]$ and hence requires a set of reference locations to be specified.
Additionally, the PSD in our approach is computed using the true sGP and is therefore invariant to the approximation choices, such as the number and placement of the knots.
Finally, we choose an exponential prior for $\sigma(h)$ that encourages $g$ to behave like a sinusoidal function, which is motivated by \cite{pcprior}.
The prior for the original $\sigma$ is then recovered by scaling the exponential prior for $\sigma(h)$ with $\alpha/\sqrt{(h/2 - \sin(2\alpha h)/4\alpha)}$.


\section{Examples}\label{sec:examples}

\subsection{Canadian Mortality}\label{sec:Mortality}

In this example, we utilize the sGP to analyze the cause-specific mortality counts in Ontario from Statistics Canada. These data were collected weekly from Jan 9th 2010 to January 1st 2022. 
To illustrate the methodology, we consider the mortality counts due to heart attack, influenza and accidental injuries. 
We fit models using data prior to March 1st, 2020 and use the remaining data to quantify the excess mortality of each cause during the COVID-19 pandemic.


For all three datasets (heart attack, influenza, and accidental injuries), we use the following model:
\begin{equation}
    \begin{aligned}
        y_i|\lambda_i &\sim \text{Poisson}(\lambda_i) ,\\
        \log(\lambda_i) &= \eta_i = \beta_0 + \beta_1 x_i + g(x_i) + \xi_i,\\
        g &\sim \text{sGP}(a = 2\pi, \sigma), \ 
        \xi_i \sim N(0,\sigma_\xi^2),
    \end{aligned}
\end{equation}
where $y_i$ denotes the mortality count and $x_i$ denotes the time of measurement in years since Jan 9th 2010.
The observation-level random effect $\xi_i$ is used to capture the overdispersion. 
We assign independent normal priors with mean 0 and variance 1000 to the fixed effects ($\beta_0$ and $\beta_1$) and the boundary conditions of the sGP.
For the variance parameters $\sigma$ and $\sigma_\xi$, we use exponential priors such that:
\begin{equation}
    \begin{aligned}
        &\mathbb{P}(\sigma(1) > 0.01) = 0.5, \
        &\mathbb{P}(\sigma_{\xi} > 1) = 0.5,
    \end{aligned}
\end{equation}
where $\sigma(1)$ is the $1$-year PSD of the sGP $g$ defined in \cref{sec:prior}.


To perform inference on the $1$-year sGP, we use the augmented space method described in \cref{subsec:sGP}.
The results of the analysis for the three mortality rates are presented in \cref{fig:ResultofInflu1}. The mortality due to influenza shows a much larger deviation from yearly periodicity compared to heart attack mortality. In particular, the peak mortality rate in 2015 is much higher than peaks in the preceding years. Consequently, the prediction uncertainty of $g$ is greater for influenza compared to for heart attack. On the other hand, the effect of seasonality in injury data is relatively small compared to that in influenza and heart attack data.
The inference results for total excess mortality during the COVID-19 period are presented in \cref{fig:ResultofInflu2}. The posteriors indicate strong evidence of negative excess mortality in influenza and likely positive excess mortality in accidental injury. However, for heart attack, there is no evidence of excess mortality in either direction.

\begin{figure}[!p]
    \centering
             \subfigure[]{
      \includegraphics[width=0.45\textwidth]{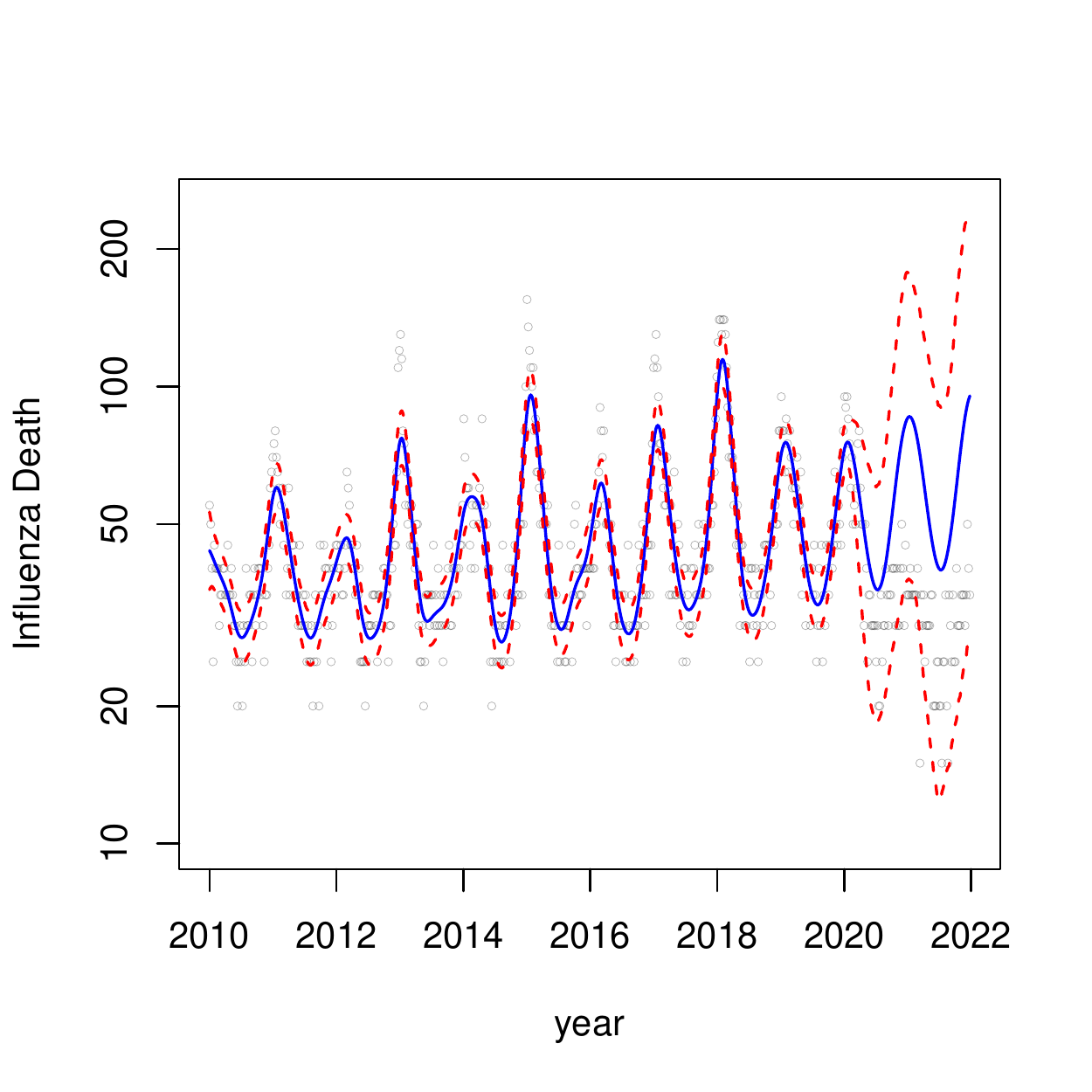}
    }
             \subfigure[]{
      \includegraphics[width=0.45\textwidth]{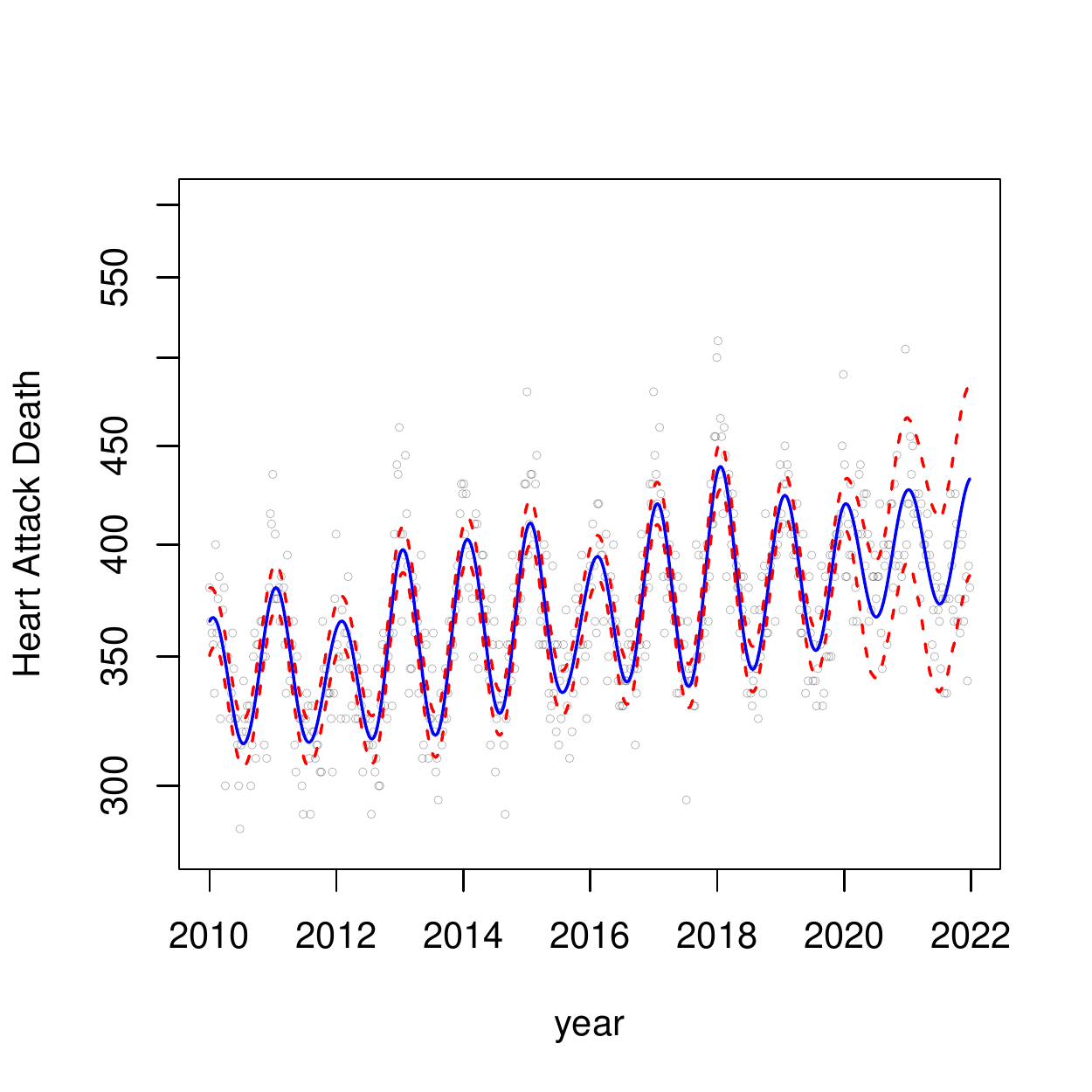}
    }
              \subfigure[]{
      \includegraphics[width=0.45\textwidth]{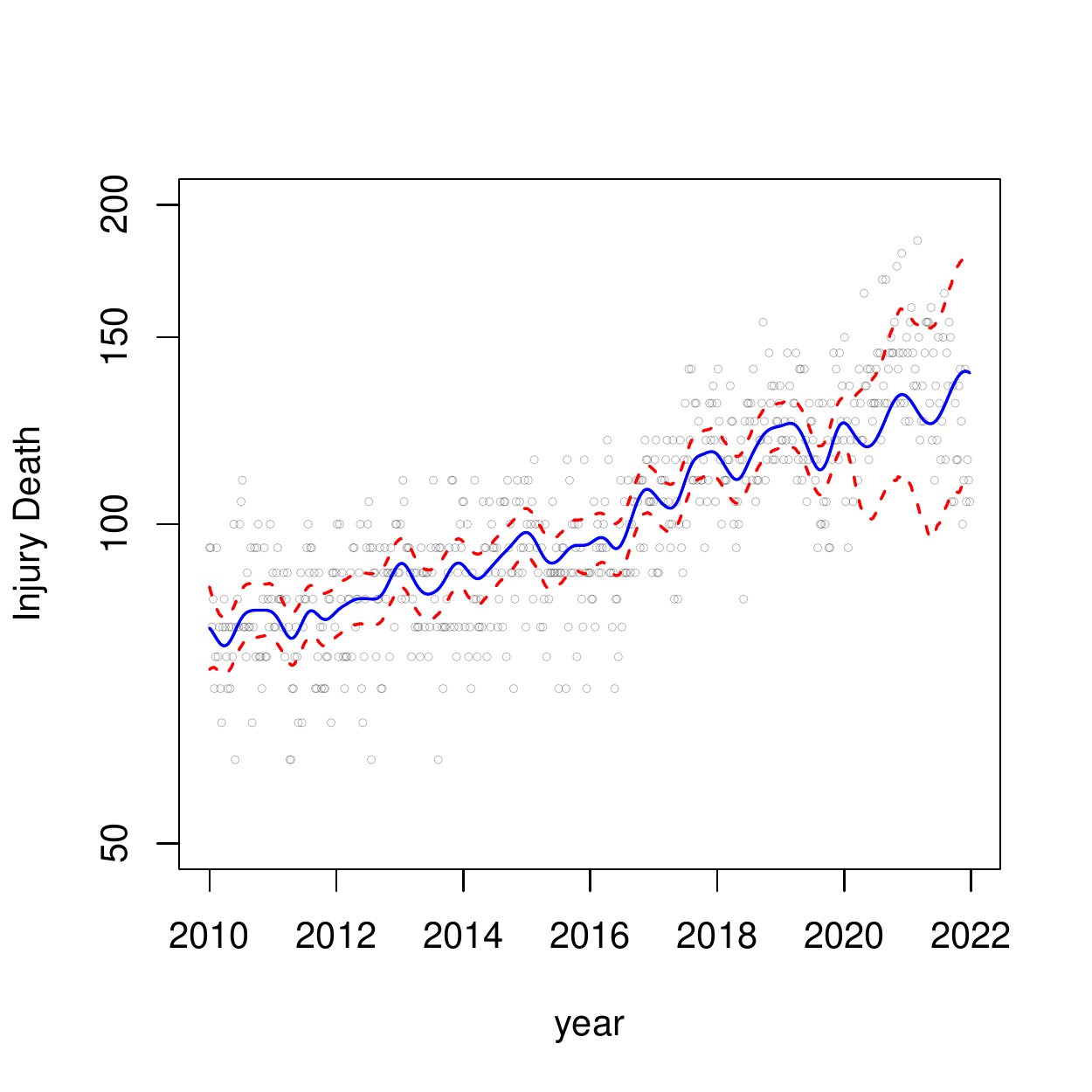}
    }
     \caption{Inferential Results for the mortality rates over time in \cref{sec:Mortality}. Figures (a)-(c) display the posterior distributions of $\exp [\beta_0 + \beta_1 x_i + g(x_i)]$ and the death counts (points). The posterior means are shown in the blue solid lines and the 95 percent posterior intervals are shown in the red dashed lines.}
    \label{fig:ResultofInflu1}
\end{figure}

\begin{figure}[!p]
    \centering
            \subfigure[]{
      \includegraphics[width=0.45\textwidth]{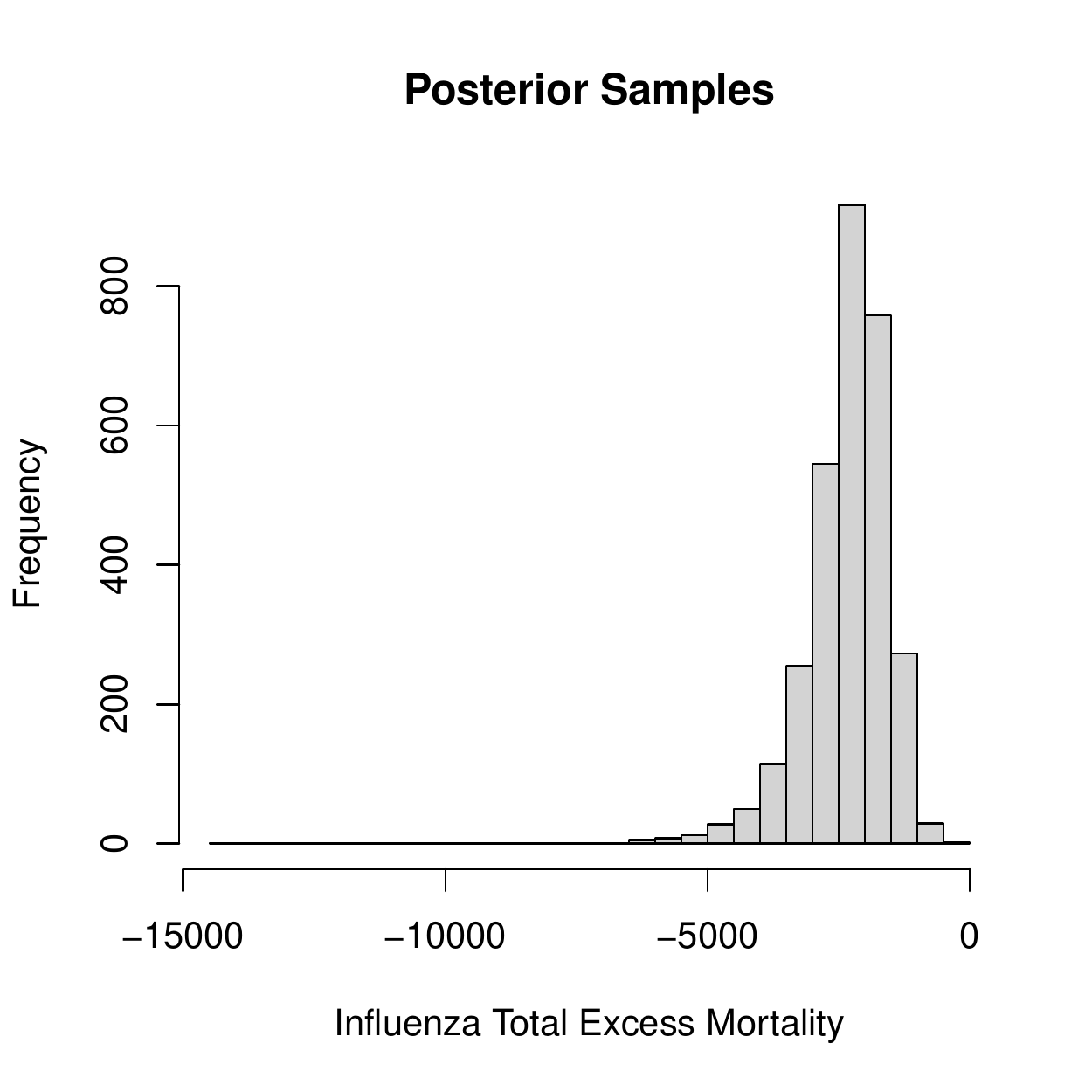}
    }
             \subfigure[]{
      \includegraphics[width=0.45\textwidth]{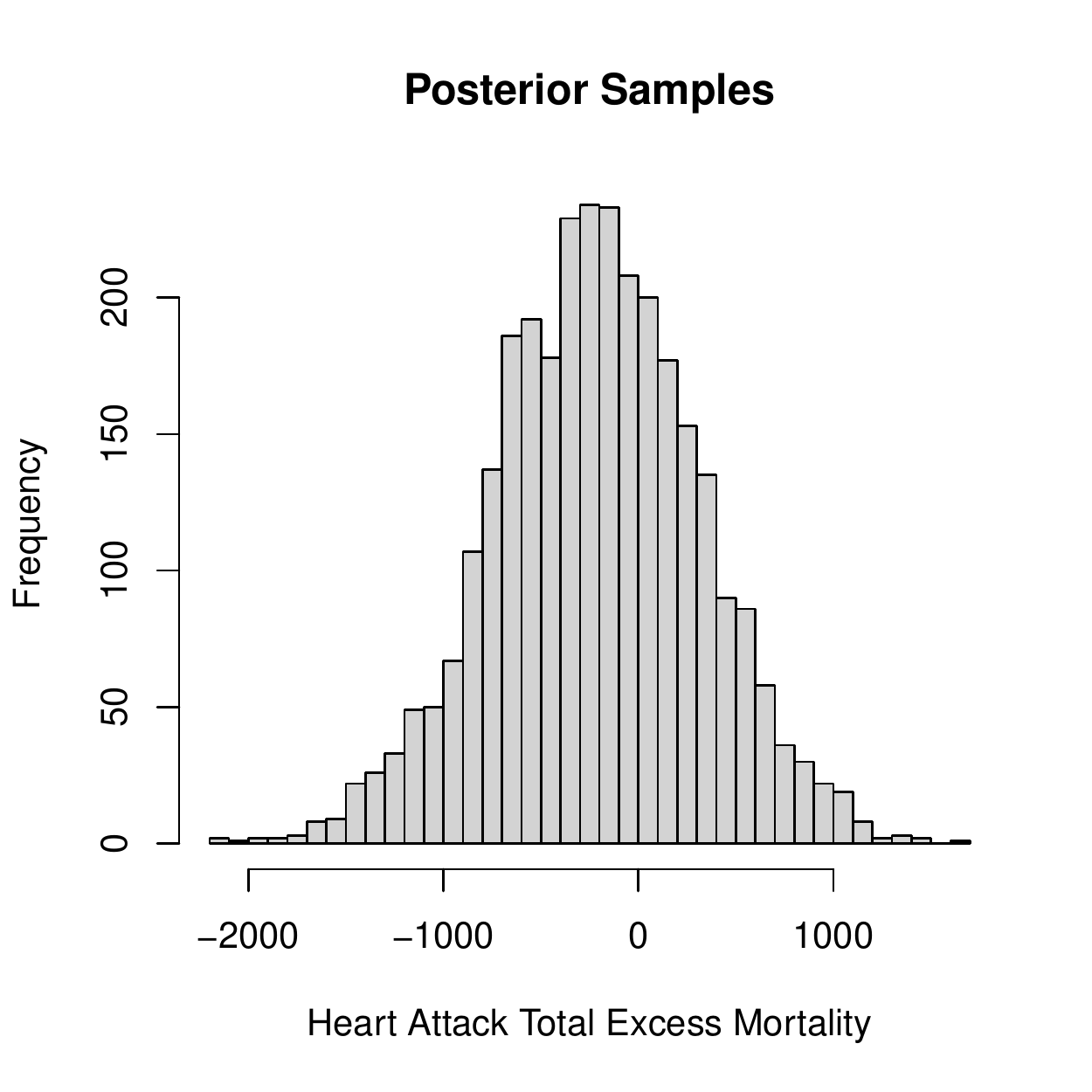}
    }
              \subfigure[]{
      \includegraphics[width=0.45\textwidth]{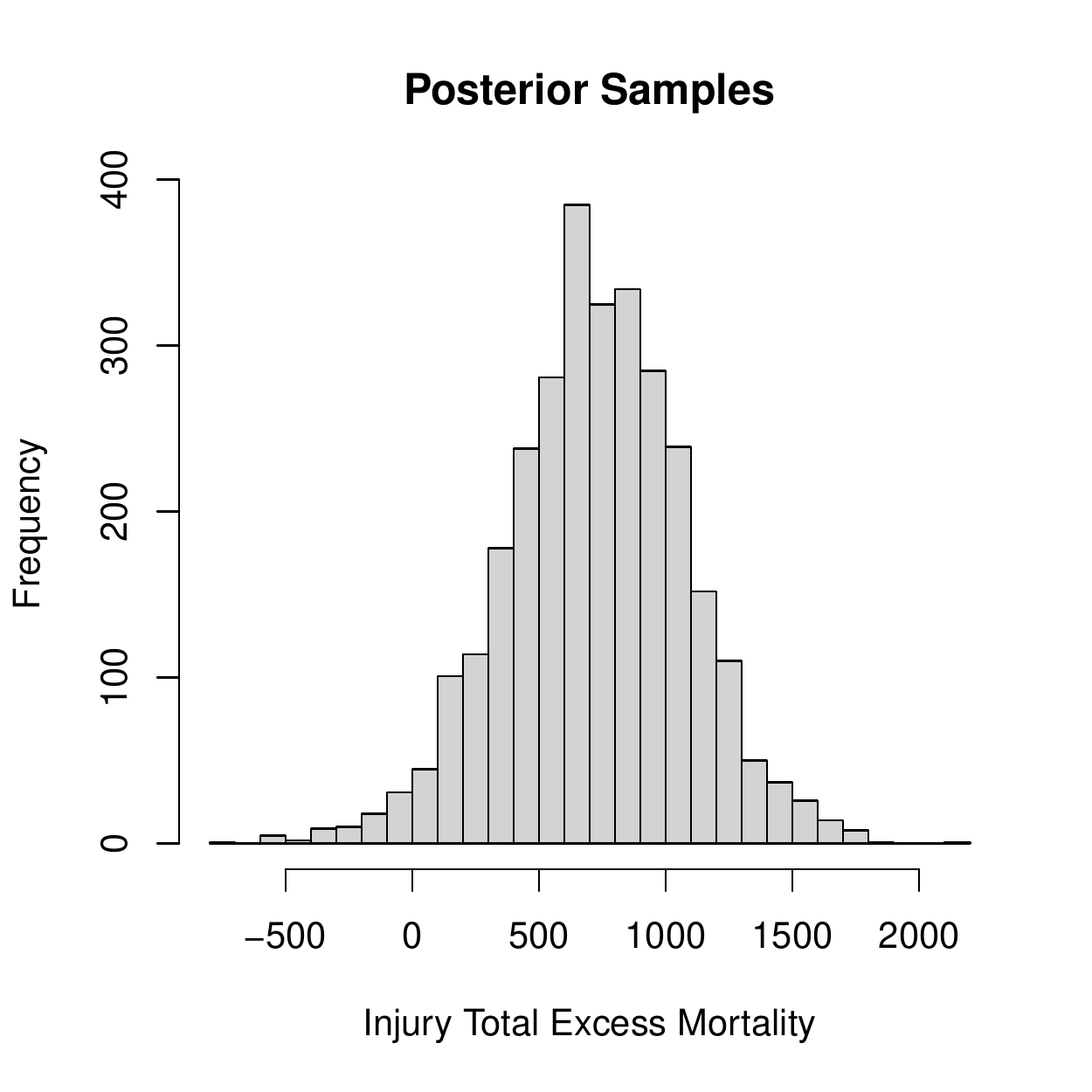}
    }
     \caption{Inferential Results for total excess mortality in \cref{sec:Mortality}. 
     Figures (a)-(c) show the posterior samples of the total excess mortality, where each sample is the sum of the difference between a posterior sample of the predicted death and the actual death counts.}
    \label{fig:ResultofInflu2}
\end{figure}

\subsection{Lynx Counts}\label{sec:Lynx}

In this section, we analyze the annual counts of trapped lynxes in Canada for 1821–1934 \citep{campbell1977survey}, with the following model:
\begin{equation}
    \begin{aligned}
        y_i|\lambda_i &\sim \text{Poisson}(\lambda_i) ,\ 
        \log(\lambda_i) = \eta_i = \beta_0 + g_{c}(x_i) + g_{c/2}(x_i) + \xi_i,\\
        g_{j} &\sim \text{sGP} \bigg(a = \frac{2\pi}{j}, \sigma_{j}\bigg)\text{ for } j \in \bigg\{c, \frac{c}{2}\bigg\}, \
        \xi_i \sim N(0,\sigma_\xi^2),
    \end{aligned}
\end{equation}
where $y_i$ denotes the counts of trapped lynxes in Canada and $x_i$ denotes the year since $1821$.
We assume the variation in lynxes counts is driven by two independent sGPs, with $c$ years ($\alpha = 2\pi/c$)  and $c/2$ years cycle ($\alpha = 4\pi/c$).
We assign independent exponential priors such that $\mathbb{P}(\sigma_j(50) > 1) = 0.01$ for each $j \in \{c, c/2\}$, and $\mathbb{P}(\sigma_\xi > 1) = 0.01$. All the boundary conditions in the sGPs and the intercept $\beta_0$ are assigned with independent normal prior $N(0,1000)$.
The year of periodicity $c$ is assumed unknown between $6$ years to $12$ years, and we place a discrete uniform prior for $c$ from $6$ to $12$ years with $0.1$ year spacing.

We carry out the inference of the sGPs with the state-space approach described in \cref{subsec:sGP}.
The inferential results are summarized in \cref{fig:ResultofLynx}. Based on \cref{fig:ResultofLynx}(a), it is most likely that the data exhibits a roughly 10.1 years cyclic variation, which is close to the duration of the solar cycle. The two sGPs well capture the quasi-periodic behaviors in the variation, as illustrated in \cref{fig:ResultofLynx}(b). 

\begin{figure}[!p]
    \centering
            \subfigure[]{
      \includegraphics[width=0.45\textwidth]{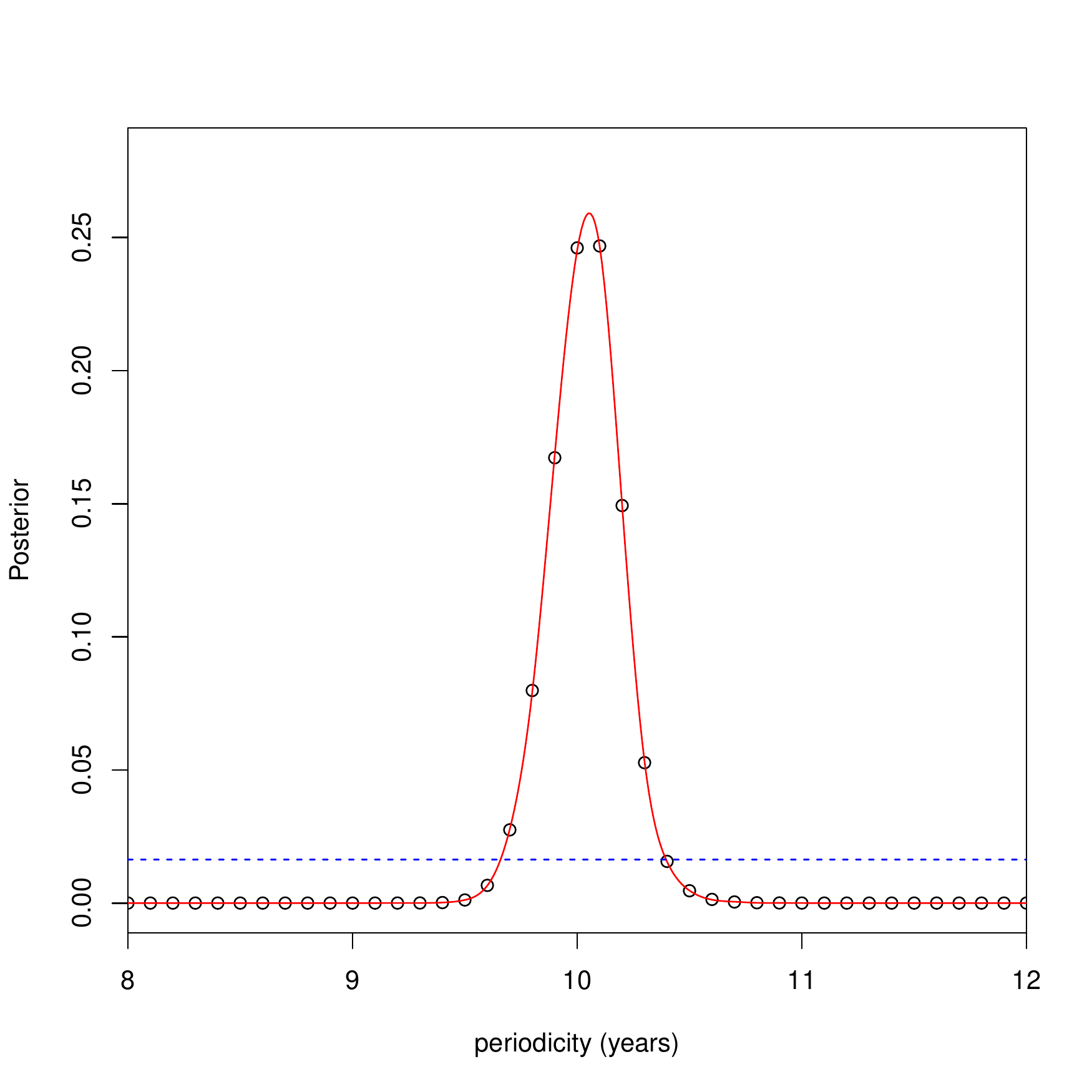}
    }
             \subfigure[]{
      \includegraphics[width=0.45\textwidth]{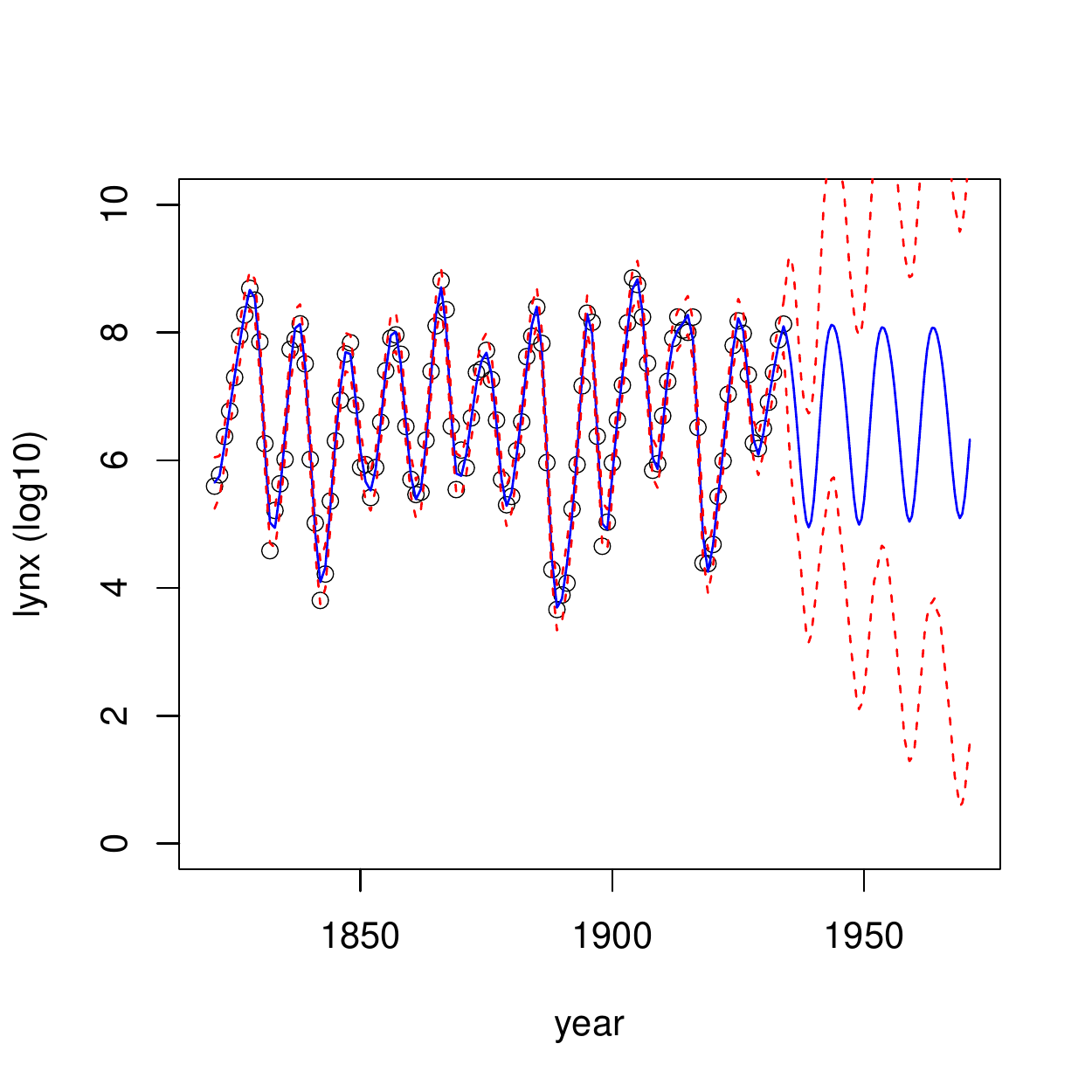}
    }
     \caption{Inferential Results for the Lynx analysis in \cref{sec:Lynx}. 
     Figure (a) shows the posterior distribution of the unknown periodicity year $c$. Figure (b) shows the posterior distribution of $g_{c} + g_{c/2}$: posterior mean shown in the blue line, 95 credible intervals shown in the red lines, and observations shown in circles. }
    \label{fig:ResultofLynx}
\end{figure}

\subsection{Sunspot Data}\label{sec:sunspot}
In this example, we analyze the yearly mean sunspot data from the World Data Center SILSO, Royal Observatory of Belgium, Brussels \citep{sidc}, with the following model:
\begin{equation}
    \begin{aligned}
         y_i &= g_{s}(x_i) + g_{tr}(x_i) + \epsilon_i,\\
        g_{s} &\sim \text{sGP}\bigg(a = \frac{2\pi}{c}, \sigma_s \bigg), \ g_{tr} \sim \text{IWP}_3(\sigma_{tr}),\\
        \epsilon_i &\sim N(0,\sigma_\epsilon^2),
    \end{aligned}
\end{equation}
where $y_i$ denotes the log-transformed yearly mean sunspot, and $x_i$ denotes the year of measurement since 1700.

We model the seasonal variation in yearly mean sunspots with an sGP model with periodicity around $c$ years, denoted as $g_s$. 
The long-term trend $g_{tr}$ is modelled with a third-order Integrated Wiener process (IWP-$3$) with smoothing parameter $\sigma_{tr}$.
For variance parameters, we use independent exponential priors such that $\mathbb{P}(\sigma_c(50) > 1) = 0.5$, $\mathbb{P}(\sigma_{tr}(50) > 5) = 0.5$ and $\mathbb{P}(\sigma_\epsilon > 1) = 0.1$. All the boundary conditions in the sGP are assigned with independent normal prior $N(0,1000)$.
The year of periodicity $c$ is assumed unknown between $8$ years to $13$ years, and we place a discrete uniform prior for $c$ from $8$ to $13$ years with $0.1$ year spacing.

To facilitate the computation, we then approximate $g_{tr}$ using 100 O-spline basis functions introduced in \citep{iwpus}, and $g_s$ using 300 sB-spline basis functions as described in \cref{sec:FEM}. 
All the basis functions are constructed with equally spaced knots.
The inferential results are summarized in \cref{fig:ResultofSunspot}. Similar to the results in \cref{sec:Lynx}, the posterior distribution of $c$ is concentrated between $9.5$ years to $11.5$ years, close to the duration of the solar cycle.

\begin{figure}[!p]
    \centering
            \subfigure[]{
      \includegraphics[width=0.45\textwidth]{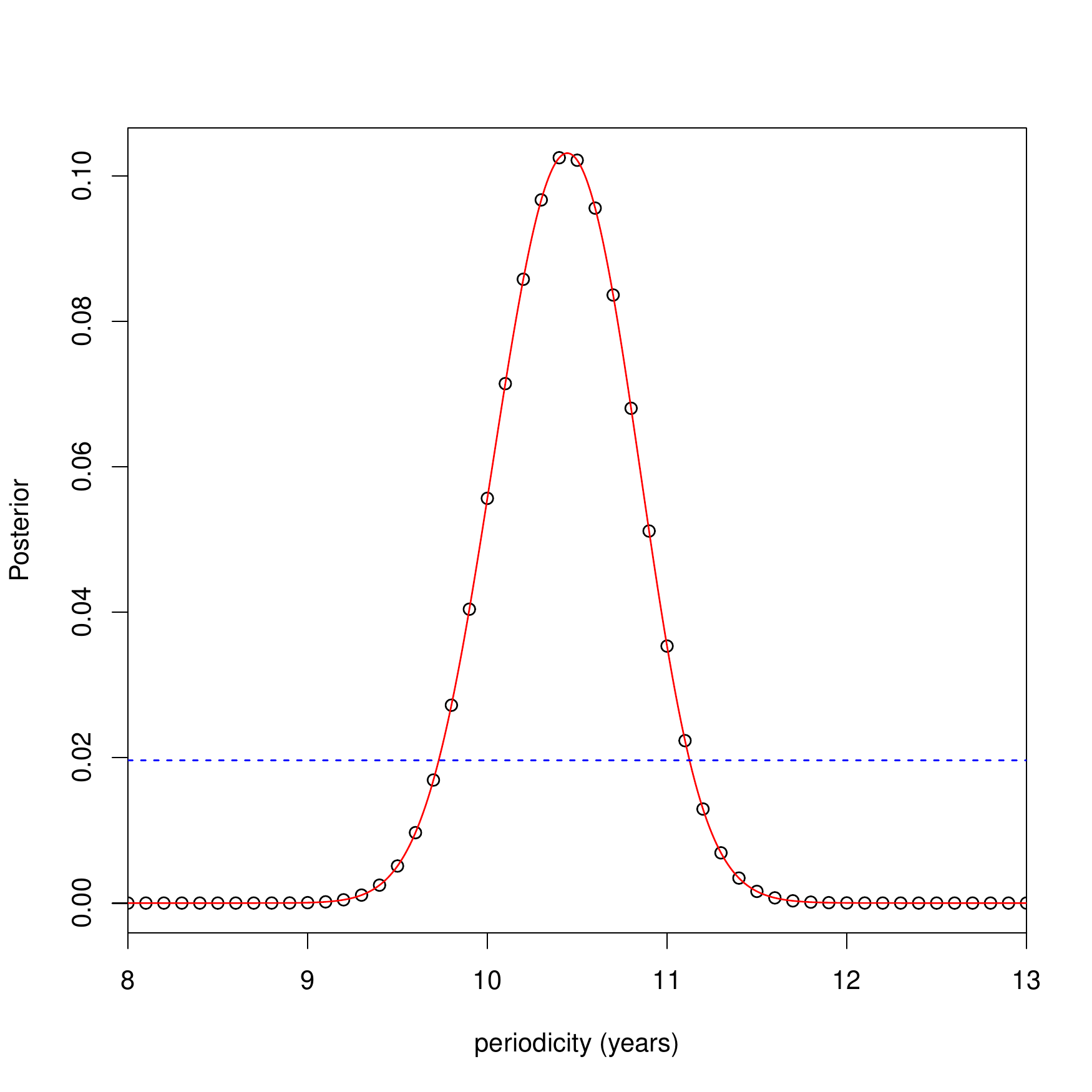}
    }
             \subfigure[]{
      \includegraphics[width=0.45\textwidth]{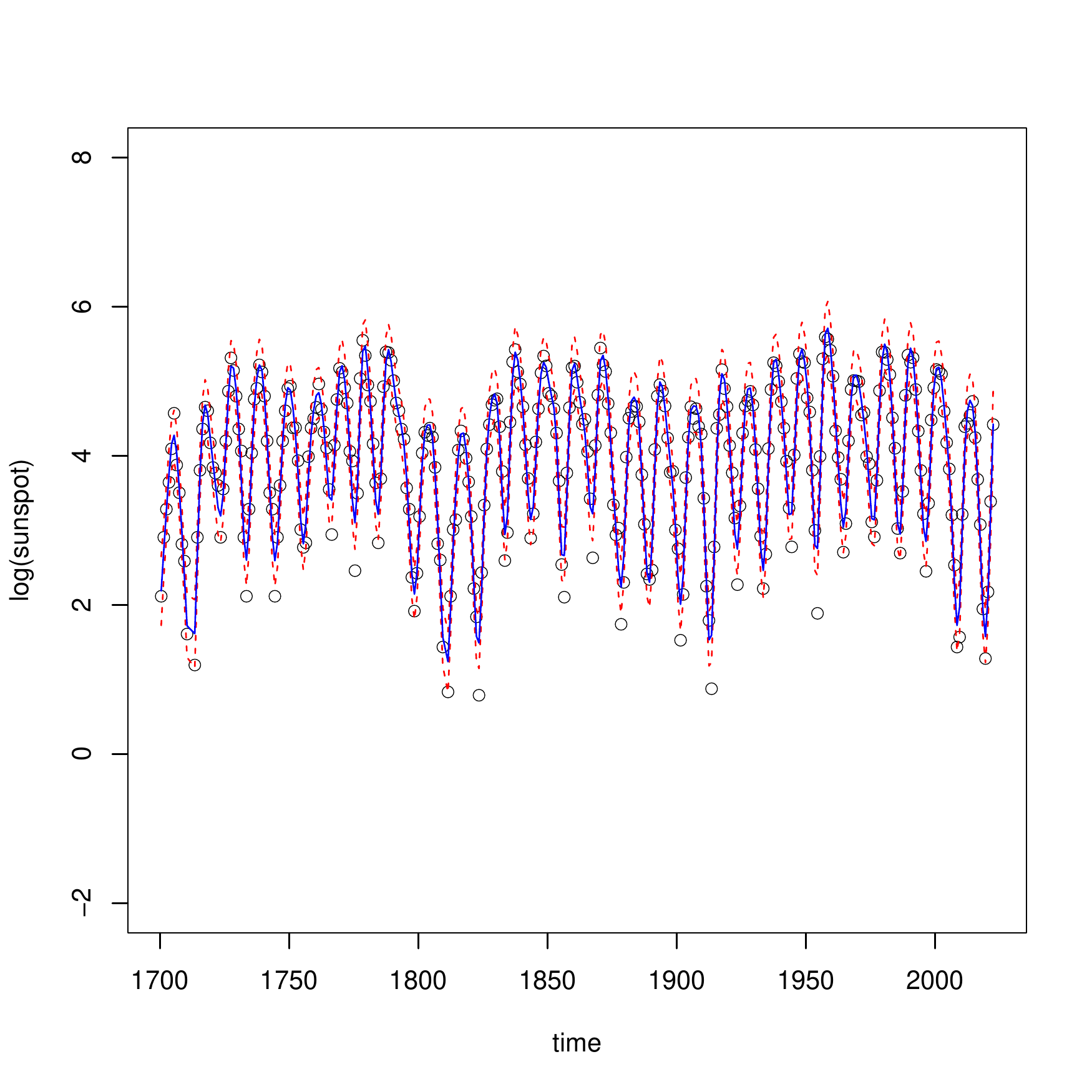}
    }
     \caption{Inferential Results for the sunspot analysis in \cref{sec:sunspot}. 
     Figure (a) shows the posterior distribution of the unknown periodicity year $c$. 
     Figure (b) shows the posterior distribution of $g_s + g_{tr}$: posterior mean shown in the blue line, 95 credible intervals shown in the red lines, and observations shown in circles.}
    \label{fig:ResultofSunspot}
\end{figure}

\subsection{CO2 Variation}\label{sec:CO2}

In this section, we will study the atmospheric Carbon Dioxide (CO2) concentrations data collected from an observatory in Hawaii from 1960 to 2021. The observations were collected monthly before May 1974 and weekly afterward.

Previous analyses have suggested the presence of several periodic components in addition to the obvious yearly cycle in the CO2 variation. 
\cite{rust1979inferences} analyzed the data from this observatory and hypothesized that the besides the obvious yearly cyclic variation, two longer periodic components also likely exist, due to the southern oscillation index (ENSO) and solar activity cycles.
Therefore, \cite{rust1979inferences} adopted a model with an exponential function for the trend, one sinusoidal function and its first harmonic component for the annual cycle, and two additional sinusoidal functions for the ENSO and solar cycles. 
In a later study of the climate oscillation, \cite{scafetta2010empirical} also identified a 10.4-year cycle induced by the solar cycle, and an additional 9.1-year cycle induced by long lunar tidal behaviors


Subsequently, both \cite{rasmussen2003gaussian} and \cite{solin2014explicit} analyzed the CO2 data with Gaussian process regressions. They assumed the CO2 concentration consists of an increasing long-term trend, a pronounced yearly seasonal variation, and some irregular noise. 
\cite{rasmussen2003gaussian} suggested that the seasonal variation in the CO2 variation is smooth but not exactly periodic, and both \cite{rasmussen2003gaussian} and \cite{solin2014explicit} chose to model the quasi-periodic seasonal variation as one GP with covariance being the product of the canonical periodic covariance and the Matern covariance with $\nu = 3/2$ or $\infty$.

Based on these existing studies, we consider the following model (M1):
\begin{equation}
    \begin{aligned}
        y_i &= g_{tr}(x_i) + g_{1}(x_i) + g_{\frac{1}{2}}(x_i) + g_{\frac{44}{12}}(x_i) + g_{9.1}(x_i) + g_{10.4}(x_i) + e_i,\\
        g_{s} &\sim \text{sGP}\bigg(a = \frac{2\pi}{s}, \sigma_{s}\bigg),\text{ for } s \in \bigg\{1, \frac{1}{2}, \frac{44}{12}, 9.1, 10.4 \bigg\} \\ 
        g_{tr} &\sim \text{IWP-3}(\sigma_{tr}), \ 
        e_i \sim N(0,\sigma_e^2),
    \end{aligned}
\end{equation}
where each $y_i$ denotes the CO2 concentration and $x_i$ denotes the time of measurement since March 30, 1960, in years.
The long-term trend $g_{tr}$ is modelled using a third-order IWP with smoothing parameter $\sigma_{tr}$. 
The seasonal variation is modelled as the sum of five separate sGPs: the first two sGPs capture the one-year cycle and its first-order harmonic, and the third sGP capture the 44-month ENSO cycles suggested in \cite{rust1979inferences}, and the fourth and fifth sGPs respectively capture the 9.1-year lunar cycle and 10.4-year solar cycle suggested in \cite{scafetta2010empirical}.
All the components in the model are considered to be independent.

In contrast to \cite{rasmussen2003gaussian} and \cite{solin2014explicit} where all the hyperparameters are estimated through maximizing marginal likelihood, we assign independent exponential priors to all the variance parameters such that:
\begin{equation}
    \begin{aligned}
        &\mathbb{P}(\sigma_{s}(10) > 1) = 0.5, \text{ for } s \in \bigg\{1, \frac{1}{2}, \frac{44}{12}, 9.1, 10.4 \bigg\}\\
        &\mathbb{P}(\sigma_{tr}(10) > 30) = 0.5, \
        \mathbb{P}(\sigma_{e} > 1) = 0.5.
    \end{aligned}
\end{equation}
All the boundary conditions in the sGP or IWP are assigned with independent priors $N(0,1000)$. 
As a comparison, we also consider an exact periodic model (M2):
\begin{equation}
    \begin{aligned}
        y_i &= g_{tr}(x_i) + g_{1}(x_i) + g_{\frac{1}{2}}(x_i) + g_{\frac{44}{12}}(x_i) + g_{9.1}(x_i) + g_{10.4}(x_i) + e_i,\\
        g_{s}(x) &= v_{s_1} \cos\bigg(\frac{2\pi x}{s}\bigg) + v_{s_2} \sin\bigg(\frac{2\pi x}{s}\bigg), \text{ for } s\in \bigg\{1, \frac{1}{2}, \frac{44}{12}, 9.1, 10.4 \bigg\} \\ 
        g_{tr} &\sim \text{IWP-3}(\sigma_{tr}), \ 
        e_i \sim N(0,\sigma_e^2),
    \end{aligned}
\end{equation}
where all the sGP components in M1 are replaced with just their corresponding sinusoidal functions, and all the priors are the same as in M1.

In order to improve computational efficiency, we adopt FEM approximations for all the $\GP$s in M1 and M2, with basis functions constructed using equally spaced knots. Specifically, we utilize the sB-splines approximation with 90 basis functions for the sGP models, and the O-Spline approximation \citep{iwpus} with 50 basis functions for the IWP model. Once the models M1 and M2 are fitted, we use them to forecast CO2 variation up to the year 2030. The inference and prediction results are summarized in \cref{fig:ResultofCO2}, \cref{fig:DetailedResultofCO2} and \cref{table:co2}.

As shown in \cref{fig:ResultofCO2}, the sGP-based model (M1) produces a much more stable prediction than the exact-periodic model (M2), demonstrated by the width of the posterior interval in the forecast. 
This is because, the value of $\sigma_{tr}(10)$ in M2 is estimated to be much larger (Post Median: 100.6) compared to in M1 (Post Median: 0.593), due to the failure of M1 to accommodate the quasi-periodic behavior of the seasonal variation. 
This difference can be further confirmed in \cref{fig:DetailedResultofCO2}, where the derivative of the long-term trend has smooth sample paths in M1 but very wiggly sample paths in M2.

\begin{table}[]
\centering
\begin{tabular}{llll}
                            & 1st Quantile    & Median          & 3rd Quantile    \\
$\sigma_{tr}(10)$             & 0.4371 (102.05)  & 0.6446 (120.81)  & 0.9462 (139.49)  \\
$\sigma_{1}(10)$             & 1.8179 (--)      & 2.4542 (--)      & 3.1709 (--)      \\
$\sigma_{\frac{1}{2}}(10)$    & 1.8179 (--)      & 2.4542 (--)      & 3.1709 (--)      \\
$\sigma_{\frac{44}{12}}(10)$  & 0.4004 (--)     & 0.4514 (--)     & 0.5019 (--)     \\
$\sigma_{9.1}(10)$ & 0.004248 (--)    & 0.01045 (--)    & 0.02445 (--)    \\
$\sigma_{10.4}(10)$ & 0.01090 (--)    & 0.02035 (--)    & 0.03335 (--)    \\
$\sigma_{\epsilon}$           & 0.5870 (0.6099) & 0.5928 (0.6159) & 0.5988 (0.6220)
\end{tabular}
     \caption{Posterior summary of variance parameters for the CO2 example in \cref{sec:CO2}. Results from M2 are shown in parenthesis. }
     \label{table:co2}
\end{table}

\begin{figure}[!p]
    \centering
             \subfigure[M1: Overall mean CO2 concentration]{
      \includegraphics[width=0.45\textwidth]{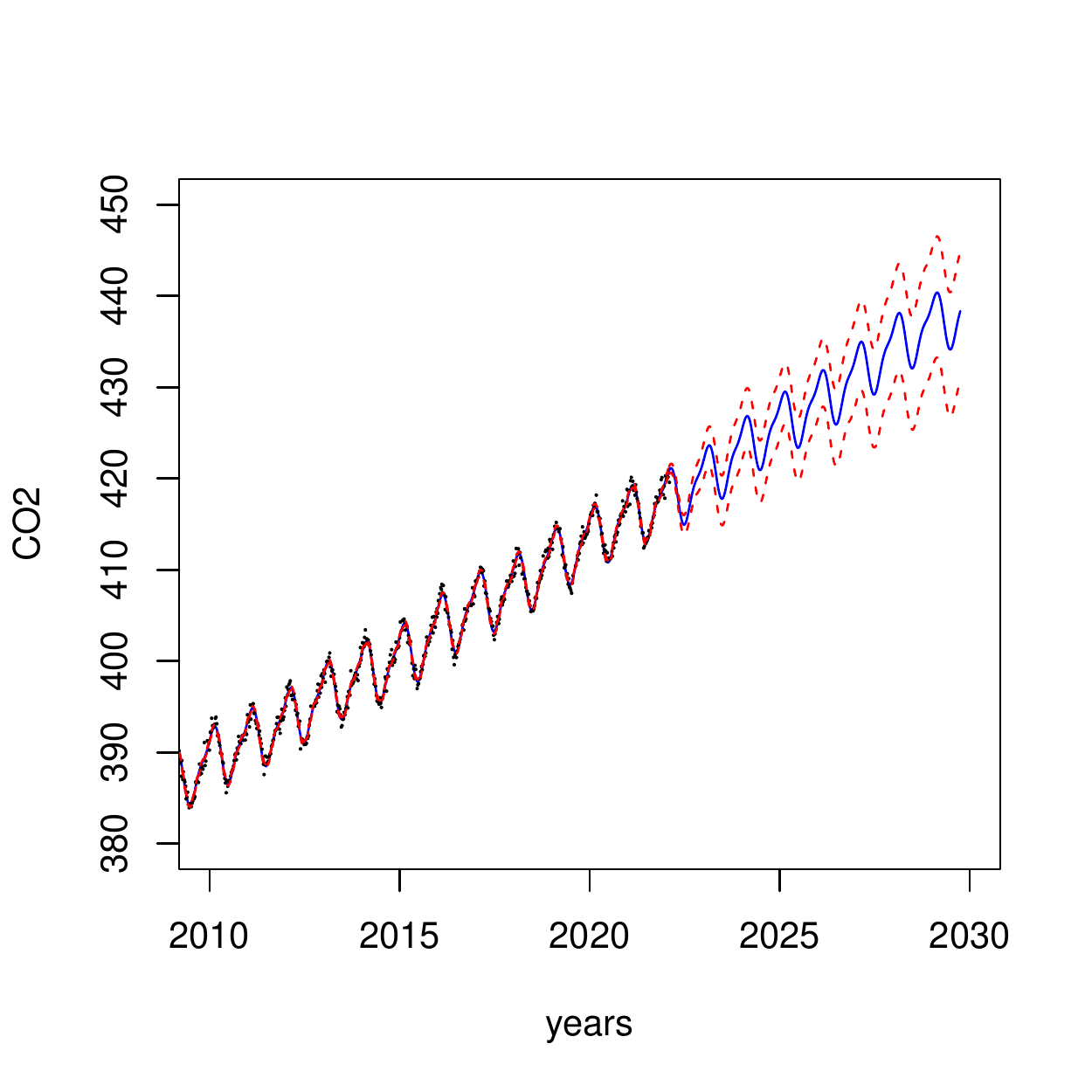}
    }
              \subfigure[M1: Seasonal variation]{
      \includegraphics[width=0.45\textwidth]{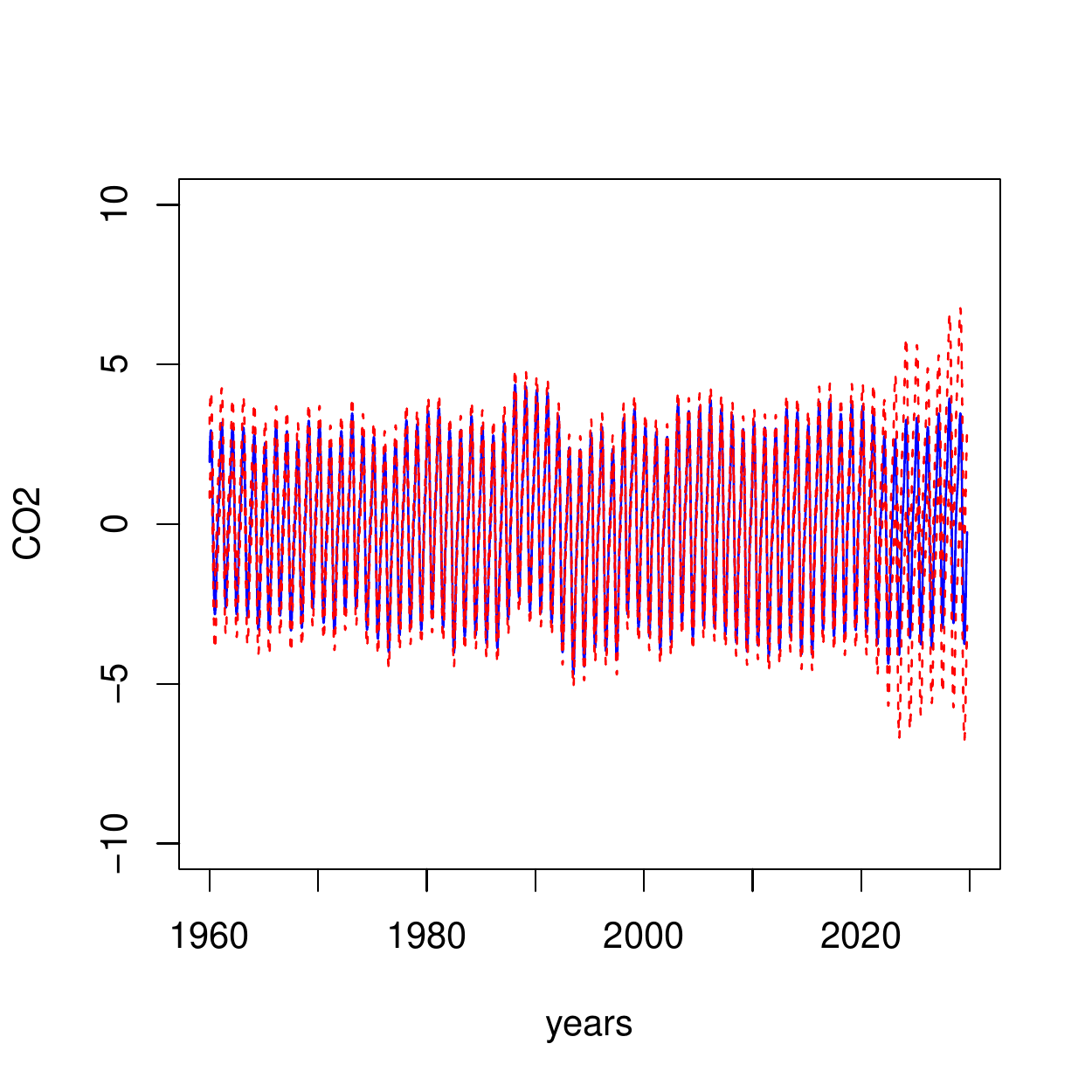}
    }
              \subfigure[M2: Overall mean CO2 concentration]{
      \includegraphics[width=0.45\textwidth]{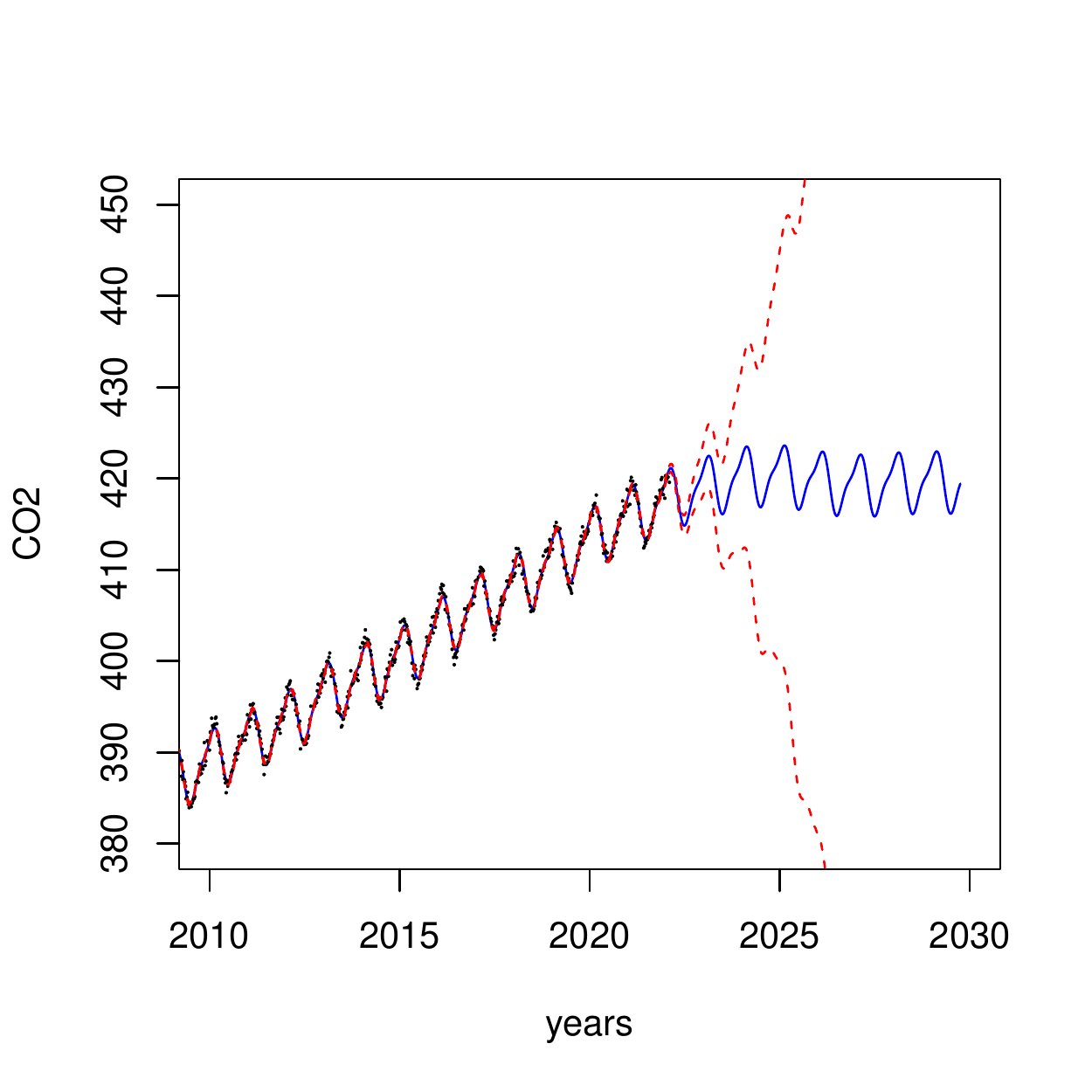}
    }
              \subfigure[M2: Seasonal variation]{
      \includegraphics[width=0.45\textwidth]{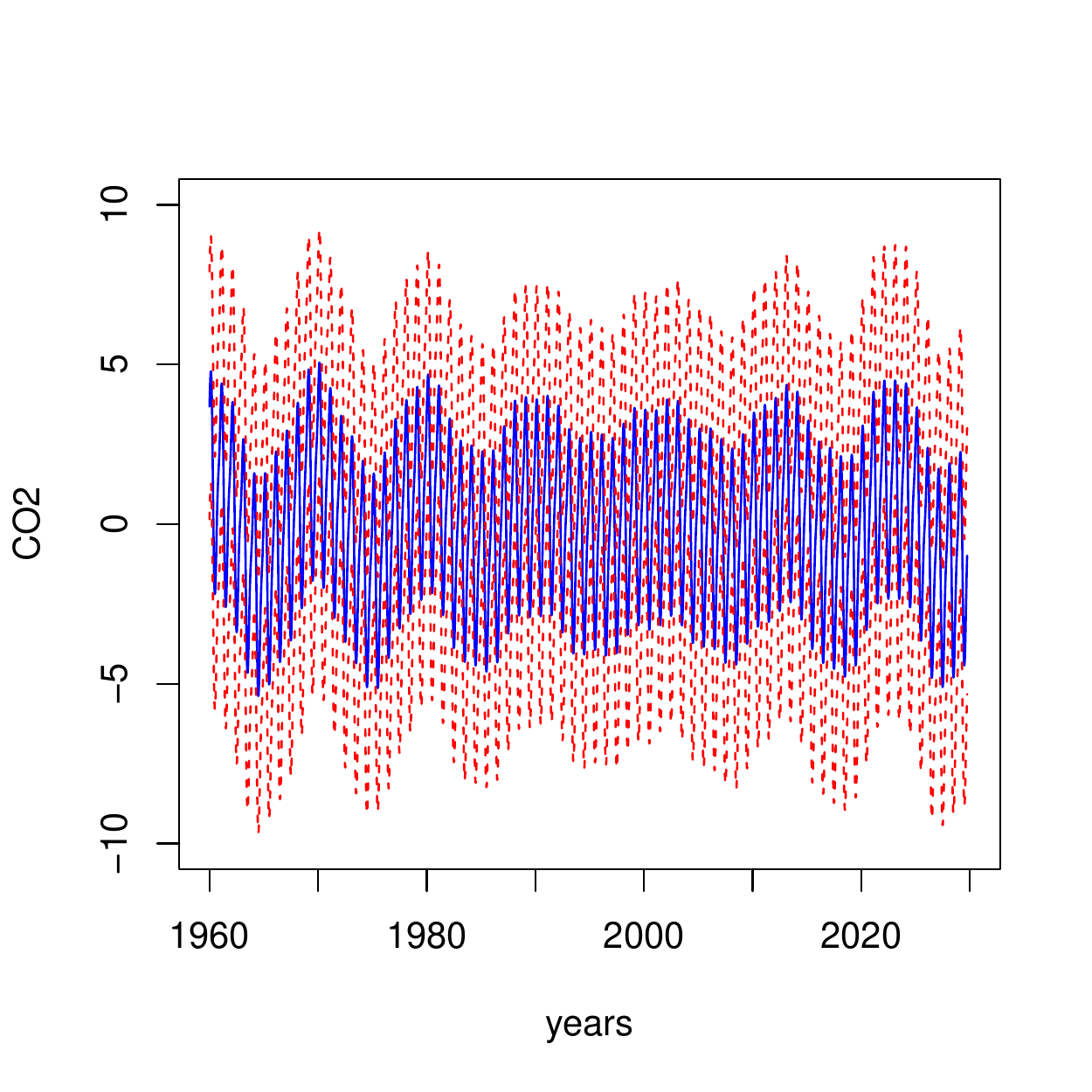}
    }
     \caption{Inferential Results for the CO2 example in \cref{sec:CO2}. (a,c) display the posterior of the overall mean CO2 concentration over years (i.e. $g_{tr} + \sum_s g_s$). (b,d) show the posterior of the overall seasonal component in the CO2 concentration (i.e. $\sum_s g_s$). 
     In (a)-(d), the blue solid lines denote the posterior mean, and the red dashed lines denote the 95 percent posterior credible interval.
     The prediction result from M1 is more stable than that from M2, as shown by the width of the interval.
     }
    \label{fig:ResultofCO2}
\end{figure}

\begin{figure}[!p]
    \centering
            \subfigure[M1: Long-term trend $g_{tr}$]{
      \includegraphics[width=0.45\textwidth]{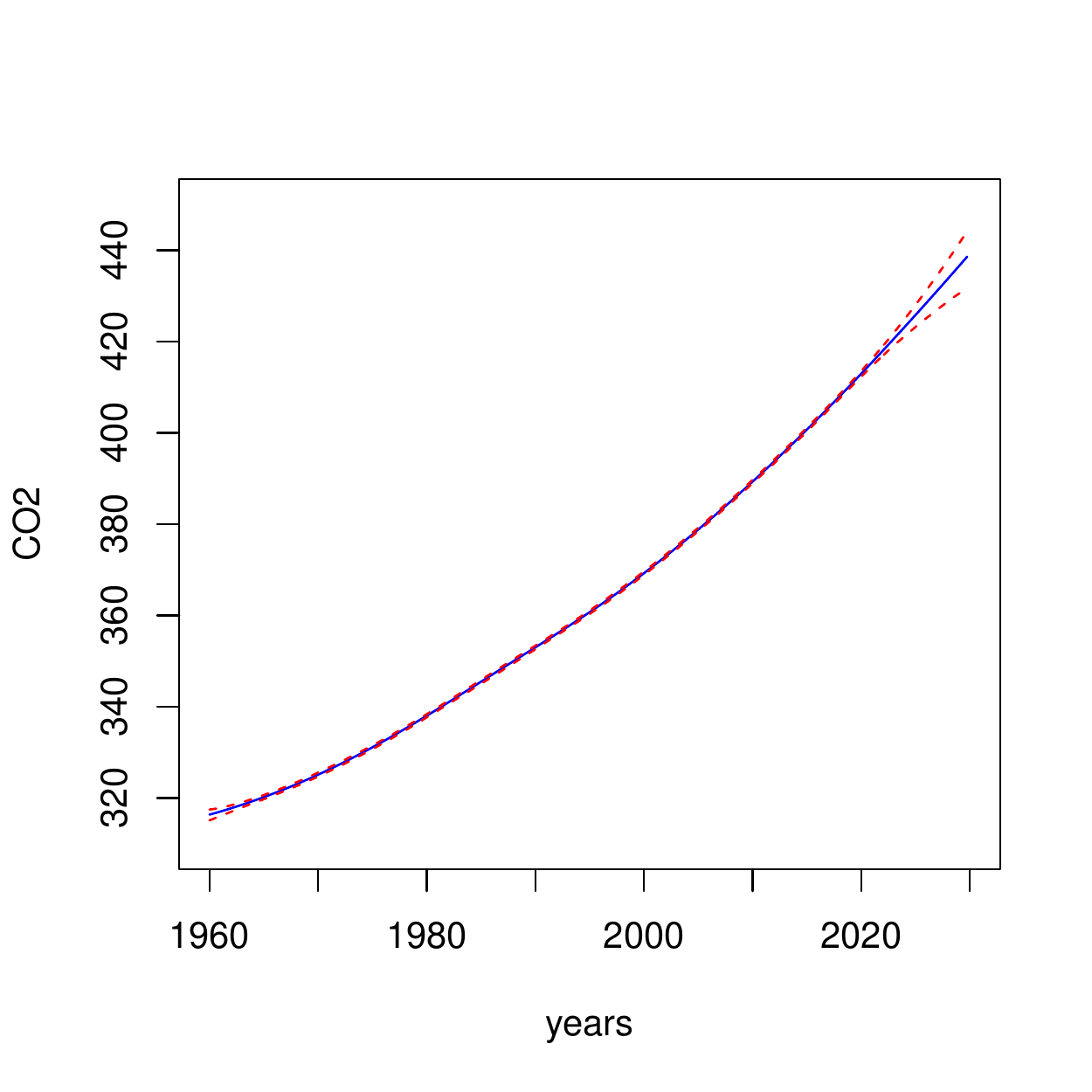}
    }
            \subfigure[M1: Inference of $g'_{tr}$]{
      \includegraphics[width=0.45\textwidth]{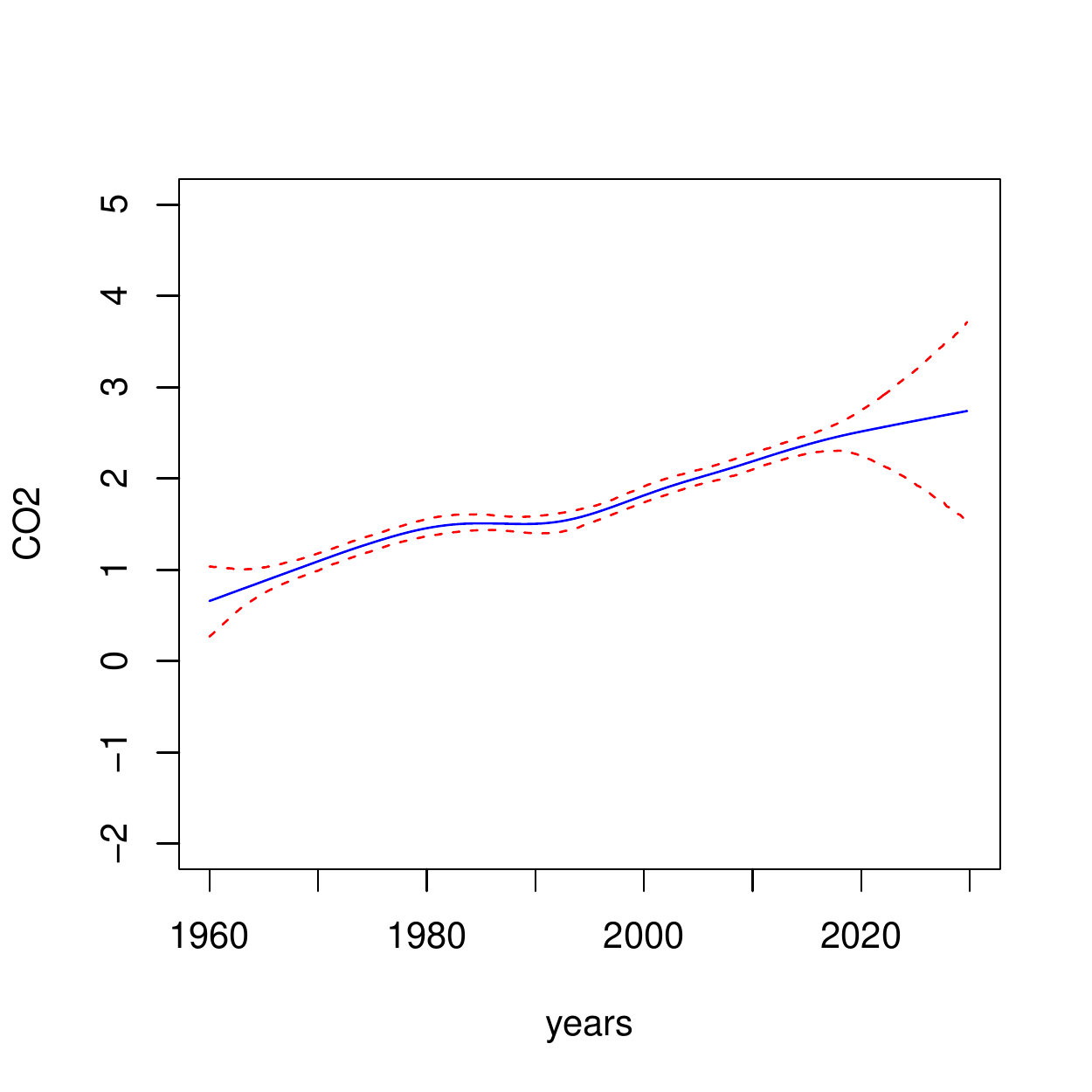}
    }
              \subfigure[M2: Long-term trend $g_{tr}$]{
      \includegraphics[width=0.45\textwidth]{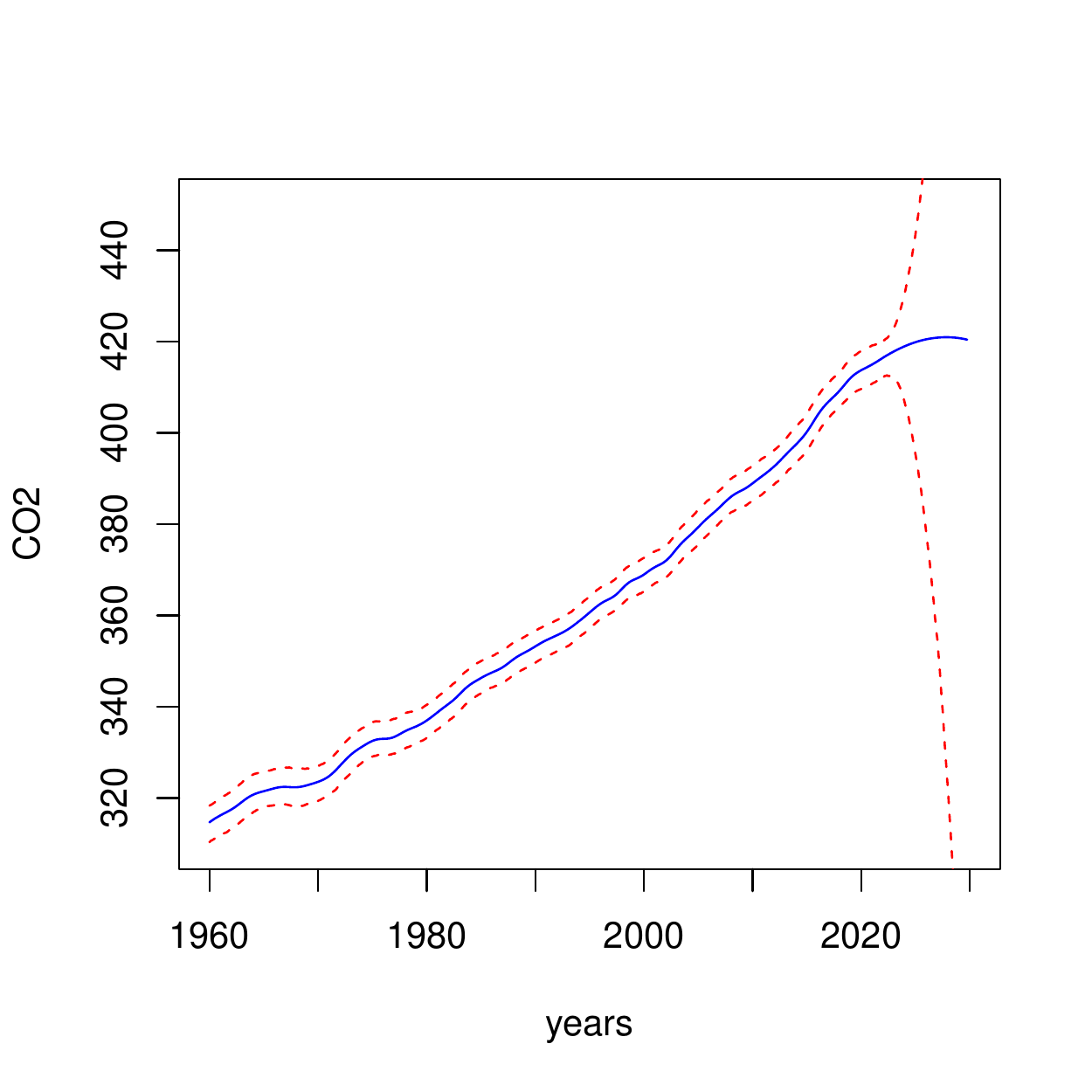}
    }
             \subfigure[M2: Inference of $g'_{tr}$]{
      \includegraphics[width=0.45\textwidth]{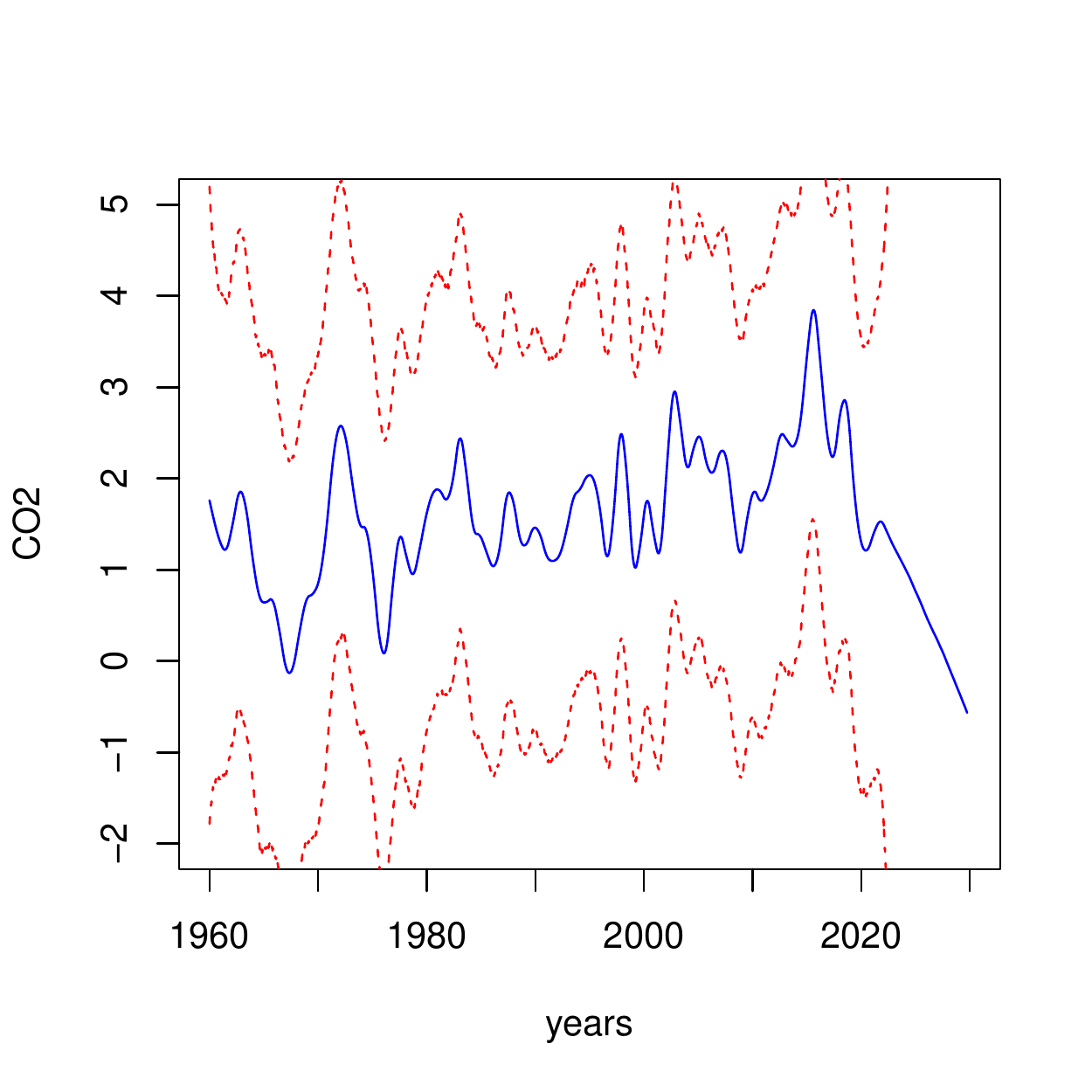}
    }
     \caption{Detailed inferential results for the CO2 example in \cref{sec:CO2}. (a,c) display the posterior distribution of the long-term trend; (b,d) display the posterior distribution of the first derivative of the long-term trend. In (a)-(d), the blue solid lines denote the posterior mean, and the red dashed lines denote the 95 percent posterior credible interval.
     The derivative of $g_{tr}$ is estimated to be much smoother in M1, whereas in M2 the result is oscillating due to the unaccounted quasi-periodicity.
     }
    \label{fig:DetailedResultofCO2}
\end{figure}

\section{Discussion}\label{sec:discussion}

In this paper, we introduced a class of Gaussian processes called seasonal Gaussian processes (sGP) for making model-based inference on quasi-periodic functions.
We then proposed an accurate and computationally efficient sB-spline approximation to the sGP, which can handle densely irregularly spaced locations. 
The proposed approximation is shown to converge to the sGP as the number of basis functions increases, both theoretically and numerically. 
We implemented our proposed approaches in a flexible and scalable manner using modern approximate Bayesian inference algorithms, as described in \cite{aghqsoftware}, and demonstrated their practical usage in diverse real data applications.

While our focus was on making fully Bayesian inference for the unknown quasi-periodic function $g$, our proposed approaches could also be implemented in a mixed effect model with parameters estimated through maximum likelihood \citep{rasmussen2003gaussian} or generalized cross-validation \citep{golub1979generalized}. 
An interesting avenue for future work is to extend our approach to incorporate adaptive modeling of the sGP, where the standard deviation parameter $\sigma$ varies over time. 
This could be achieved using a strategy similar to the one described in \cite{yue2014bayesian}.

We have made the software that implements our proposed methods available on a Github repository, which can be found at \href{https://github.com/AgueroZZ/sGPfit}{https://github.com/AgueroZZ/sGPfit}.

\newpage
\begin{center}
{\large\bf SUPPLEMENTARY MATERIAL}
\end{center}

\begin{description}


\item[R-package:] The R-package \textit{sGPfit} that implements the proposed methodology in the paper can be found at \href{https://github.com/AgueroZZ/sGPfit}{github.com/AgueroZZ/sGPfit}.

\item[Codes and Data:] All the data used in this paper are publicly available for download. The mortality data in \cref{sec:Mortality} is available at \href{https://www.statcan.gc.ca/}{statcan.gc.ca}. 
The lynx data in \cref{sec:Lynx} is available in the R package \textit{datasets} (version 4.2.1).
The sunspot data in \cref{sec:sunspot} is available at \href{https://www.sidc.be/silso/}{sidc.be/silso}.
The CO2 data in \cref{sec:CO2} can be accessed from \href{https://scrippsco2.ucsd.edu}{scrippsco2.ucsd.edu}.

The codes to replicate each of the examples can be found at \href{https://github.com/AgueroZZ/sGPcode}{github.com/AgueroZZ/sGPcode}.

\item[Supplement:] The online supplement includes an R-markdown tutorial that provides step-by-step guidance for users to implement our proposed approach using the Lynx dataset in R.

\end{description}

\nocite{sarkka2019applied}
\nocite{schultz1969approximation}
\bibliographystyle{apalike}
\bibliography{bibliography}

\newpage


\section*{Appendix A: Derivation of the sGP covariance}

\begin{propS}[Covariance Function of the Seasonal Gaussian Process]
Let $g \sim \text{sGP}(\alpha, \sigma)$. 
Then $g$ has a covariance function:
\begin{equation}
\begin{aligned}
        C(x_1,x_2) &= \bigg(\frac{\sigma}{\alpha}\bigg)^2 \bigg[\frac{x_1}{2}\cos(\alpha(x_2-x_1)) - \frac{\cos(\alpha x_2)\sin(\alpha x_1)}{2\alpha} \bigg] \\
        &= \bigg(\frac{\sigma}{\alpha}\bigg)^2 \bigg[ \frac{\cos(\alpha x_2)x_1}{2}\cos(\alpha x_1) + \left(\frac{\sin(\alpha x_2)x_1}{2} - \frac{\cos(\alpha x_2)}{2\alpha}\right)\sin(\alpha x_1) \bigg],
\end{aligned}
\end{equation}
for any $x_1,x_2 \in \mathbb{R}^+$ such that $x_1\leq x_2$.
\end{propS}

\begin{proof}

It is obvious that the differential operator $L$ is linear.
Define $\boldsymbol{g}_{aug}(x) = (g(x), g'(x))^T$ and therefore $\boldsymbol{g'}_{aug}(x) = (g'(x), g''(x))^T$, then the SDE can be rewritten in the vector form:
\begin{equation}\label{equ:vector-SDE}
    \begin{aligned}
    \boldsymbol{g'}_{aug} = \boldsymbol{F} \boldsymbol{g}_{aug} + \boldsymbol{L}W,
    \end{aligned}
\end{equation}
where $\boldsymbol{F} = \begin{bmatrix} 0 & 1 \\ -\alpha^2 & 0\\ \end{bmatrix}$ and $\boldsymbol{L} = \begin{bmatrix} 0 \\ \sigma \\ \end{bmatrix}$.

\

\noindent Using the result from \cite{sarkka2019applied} (section 4.3), the solution of the linear SDE can be written as:
\begin{equation}\label{equ:linear-SDE-sol}
    \begin{aligned}
    \boldsymbol{g}_{aug}(x) &= \exp(\boldsymbol{F}x)\boldsymbol{g}_{aug}(0) + \int_0^x \exp(\boldsymbol{F}(x-\tau))\boldsymbol{L}W(\tau)d\tau \\
    &= \int_0^x \exp(\boldsymbol{F}(x-\tau)) \boldsymbol{L}W(\tau)d\tau,
    \end{aligned}
\end{equation}
where $\exp(\boldsymbol{F}x)$ denotes the matrix exponential defined as $\exp(\boldsymbol{F}x) = \sum_k\frac{\boldsymbol{F}^k x^k}{k!}$.

\

\noindent Note that $\boldsymbol{F}^{2k} = (-\alpha^2)^k \boldsymbol{I}$ and $\boldsymbol{F}^{2k} = (-\alpha^2)^k \boldsymbol{F}$. With Taylor series, the first component of $\boldsymbol{g}_{aug}(x)$ can be therefore written as:
\begin{equation}\label{equ:linear-SDE-sol-2}
    \begin{aligned}
    {g}(x) = \int_0^x \frac{\sigma}{\alpha} \sin(a(x-\tau)) W(\tau)d\tau.
    \end{aligned}
\end{equation}
Assume arbitrary $0<x_1\leq x_2$, the covariance function can be computed for $g$ as:
\begin{equation}\label{equ:true-cov}
    \begin{aligned}
    k(x_1,x_1) &= \int_0^{x_1} \frac{\sigma}{\alpha} \sin(\alpha(x_1-\tau)) \frac{\sigma}{a} \sin(\alpha(x_2-\tau)) d\tau \\
           &= \bigg(\frac{\sigma}{\alpha}\bigg)^2 \bigg[\frac{x_1}{2}\cos(\alpha(x_2-x_1)) - \frac{\cos(\alpha x_2)\sin(\alpha x_1)}{2\alpha} \bigg],
    \end{aligned}
\end{equation}
using properties of Gaussian white noise \citep{harvey1990forecasting}.

\end{proof}

\section*{Appendix B: Proof of the State-Space Representation}

\begin{theoremS}[State Space Representation of the sGP]
Consider $g \sim \text{sGP}(\alpha, \sigma)$, and let $\boldsymbol{s} = \{s_1,..., s_n\} \subset \R^{+}$ denotes a set of $n$ sorted locations and spacing $d_1 = s_1$ and $d_i = s_i - s_{i-1}$ for $i \in \{2,..,n\}$. Then the augmented vector $\boldsymbol{g}_{aug}(s_i) = [g(s_i), g'(s_i)]^T$ can be written as a Markov model:
\begin{equation}
\boldsymbol{{g}}_{aug}(s_{i+1}) = \boldsymbol{R}_{i+1} \boldsymbol{{g}}_{aug}(s_{i}) + \boldsymbol{\epsilon}_{i+1},
\end{equation}
where $\boldsymbol{\epsilon}_i \overset{ind}{\sim}N(0,\boldsymbol{\Sigma}_i)$. The $2 \times 2$ matrices $\boldsymbol{R}_{i}$ and $\boldsymbol{\Sigma}_{i} = \boldsymbol{Q}_{i}^{-1}$ are respectively defined as:
\begin{equation}
    \begin{aligned}
    \boldsymbol{R}_{i} = \begin{bmatrix} \cos(\alpha d_i) & \frac{1}{\alpha}\sin(\alpha d_i) \\
    -\alpha\sin(\alpha d_i) & \cos(\alpha d_i)  \\
    \end{bmatrix}, \  
    \boldsymbol{\Sigma}_i = \sigma^2 \begin{bmatrix} \frac{1}{\alpha^2}\bigg(\frac{d_i}{2} - \frac{\sin(2\alpha d_i)}{4\alpha}\bigg) & \frac{\sin^2(\alpha d_i)}{2\alpha^2} \\
    \frac{\sin^2(\alpha d_i)}{2\alpha^2} & \frac{2\alpha d_i + \sin(2\alpha d_i)}{4\alpha}  \\
    \end{bmatrix}.
    \end{aligned}
\end{equation}
\end{theoremS}

\begin{proof}

To show the above Markov representation, note that the value of $g(s_{i+1})$ given $g(s_i)$ can be written similarly as \citep{sarkka2019applied}: $$\boldsymbol{g}_{aug}(s_{i+1}) = \exp(\boldsymbol{F}d_{i+1})\boldsymbol{g}_{aug}(s_i) + \int_{s_i}^{s_{i+1}} \exp(\boldsymbol{F}(s_{i+1}-\tau))\boldsymbol{L}W(\tau)d\tau.$$

Recall that $\boldsymbol{F}^{2k} = (-a^2)^k \boldsymbol{I}$ and $\boldsymbol{F}^{2k+1} = (-a^2)^k \boldsymbol{F}$, then apply the Taylor series expansion for both components in the integral above. It then can be rewritten as:
\begin{equation}\label{markov}
    \begin{aligned}
    \boldsymbol{g}_{aug}(s_{i+1}) &= \exp(\boldsymbol{F}d_{i+1})\boldsymbol{g}_{aug}(s_i) + \int_{s_i}^{s_{i+1}} \exp(\boldsymbol{F}(s_{i+1}-\tau))\boldsymbol{L}W(\tau)d\tau \\
    &= \boldsymbol{R}_{i+1} \boldsymbol{g}_{aug}(s_{i}) + \int_{s_i}^{s_{i+1}} \begin{bmatrix} \frac{1}{a}\sin\big(a(s_{i+1} - \tau)\big) \\ \cos\big(a(s_{i+1} - \tau)\big) \end{bmatrix} \sigma W(\tau) d\tau \\
    &:= \boldsymbol{R}_{i+1} \boldsymbol{g}_{aug}(s_{i}) + \boldsymbol{\epsilon}_{i+1}.
    \end{aligned}
\end{equation}

Note that since each $\boldsymbol{\epsilon}_{i+1}$ involves integration at disjoint intervals, their independence follows from the property of Gaussian white noise \citep{harvey1990forecasting}. To check its covariance matrix $\boldsymbol{\Sigma_{i+1}}$, note that:
\begin{equation}\label{noise-cov}
    \begin{aligned}
    \boldsymbol{\Sigma_{i+1}} &= \sigma^2
    \resizebox{.8\textwidth}{!}{$\begin{bmatrix}
    \int_{s_i}^{s_{i+1}} \frac{1}{a^2}\sin^2\big(a(s_{i+1} - \tau)\big) d\tau  & \frac{1}{a}\int_{s_i}^{s_{i+1}} \sin\big(a(s_{i+1} - \tau)\big)\cos\big(a(s_{i+1} - \tau)\big) d\tau \\
    \frac{1}{a}\int_{s_i}^{s_{i+1}} \sin\big(a(s_{i+1} - \tau)\big)\cos\big(a(s_{i+1} - \tau)\big) d\tau &
    \int_{s_i}^{s_{i+1}} \cos^2 \big(a(s_{i+1} - \tau)\big)
    \end{bmatrix}$} \\
    & = \sigma^2
     \begin{bmatrix} \frac{1}{a^2}\bigg(\frac{d_{i+1}}{2} - \frac{\sin(2ad_{i+1})}{4a}\bigg) & \frac{\sin^2(ad_{i+1})}{2a^2} \\
    \frac{\sin^2(ad_{i+1})}{2a^2} & \frac{2ad_{i+1} + \sin(2ad_{i+1})}{4a}  \\
    \end{bmatrix},
    \end{aligned}
\end{equation}
which completes the proof.

\end{proof}

\section*{Appendix C: Finite Element Method}
The Finite Element Method (FEM) used to construct the finite-dimensional approximation can be understood as the following procedures. 

Given the (linear) stochastic differential equation (SDE) that defines the sGP model:
$$L g(x) = \sigma \xi(x),$$ where $L = a^2 + \frac{d^2}{dx^2}$ is a linear differential operator and $\xi(x)$ is the standard Gaussian white noise process.
Let $\Omega \subset \mathbb{R}^+$ denotes a bounded interval of interest.
Let $\mathbb{B}_k:=\{\varphi_i, i \in [k] \}$ denote the set of $k$ pre-specified basis functions, and let $\mathbb{T}_q:=\{\phi_i, i \in [q]\}$ denote the set of $q$ pre-specified test functions. We consider finite dimensional approximation with form $\tilde{g}(.) = \sum_{i=1}^{k}w_i \varphi_i(.)$. The weights $\boldsymbol{w} := [w_1,...,w_k]^T \in \mathbb{R}^k$ is a set of random weights to be determined. 

In our FEM construction, we used the sB-splines defined over $\Omega$ as the basis functions, and chose the test functions by $\mathbb{T}_k:=\{\phi_i = L\varphi_i, i \in [k]\}$, which is called a least squares approximation in \cite{spde}.
The distribution of the unknown weight vector can be found by fulfilling the weak formulation at the test function spaces $\mathbb{T}_k$, such that 
\begin{equation}
\langle L\tilde{g}(x) , \phi_i(x)\rangle \overset{d}= \sigma \langle \xi(x) , \phi_i(x)\rangle,
\end{equation}
for any test function $\phi_i \in \mathbb{T}_k$.
This equation can also be vectorized as:
$$\langle L\tilde{g}(x) , \phi_i(x) \rangle_{i=1}^k =H \boldsymbol w,$$ where the $ij$ component of the $k \times k$ $H$ matrix can be computed as $H_{ij} = \langle L\varphi_j(x) , L\varphi_i(x) \rangle_{i=1}^k.$

The inner product on the right $\langle \xi(x),\phi_i(x) \rangle_{i=1}^k$ will have Gaussian distribution with zero mean vector and covariance matrix $H$ by properties of Gaussian white noise \citep{harvey1990forecasting}. Therefore, the basis coefficients $\boldsymbol{w}$ will be multivariate Gaussian with zero mean and covariance $H^{-1}HH^{-1} = H^{-1}$. Each element of the matrix $H$ can be written as:
\begin{equation}\label{equ:H-expansion}
\begin{aligned}
        H_{ij} &=  \langle L\varphi_j , L\varphi_i \rangle \\
           &=  \langle a^2\varphi_j + \frac{d^2\varphi_j}{dx^2} , a^2\varphi_i + \frac{d^2\varphi_i}{dx^2} \rangle \\
           &= a^4 \langle \varphi_j , \varphi_i \rangle + a^2 \langle \frac{d^2\varphi_j}{dx^2} , \varphi_i \rangle  + a^2 \langle \varphi_j , \frac{d^2\varphi_i}{dx^2} \rangle + \langle \frac{d^2\varphi_j}{dx^2} , \frac{d^2\varphi_i}{dx^2} \rangle,
\end{aligned}
\end{equation}
hence $H = a^4 G + C + a^2 M$ with $G_{ij} =  \langle \varphi_i , \varphi_j \rangle$, $C_{ij} = \langle \frac{d^2\varphi_i}{dx^2} , \frac{d^2\varphi_j}{dx^2} \rangle$ and $M_{ij} = \langle \varphi_i, \frac{d^2\varphi_j}{dx^2} \rangle + \langle \frac{d^2\varphi_i}{dx^2}, \varphi_j \rangle$ for each element of the matrices.

\section*{Appendix D: Proof of the Convergence Result}
\begin{theorem*}[Convergence of B-spline Approximation]
Let $\Omega = [a,b]$ where $a,b \in \mathbb{R}^+$ and let $g \sim \text{sGP}(\alpha, \sigma)$. Assume $\mathbb{B}_k$ is a set of $k$ cubic B-splines constructed with equally spaced knots over $\Omega$, and $\tilde{g}_k$ denotes the corresponding FEM approximation, then:
$$\lim_{k\to \infty} \C_{k}(x_1,x_2) = \C(x_1,x_2),$$
for any $x_1,x_2 \in \Omega$, where $\C(x_1,x_2) = \Cov[g(x_1),g(x_2)]$, $\C_{k}(x_1,x_2) = \Cov[\tilde{g}_k(x_1),\tilde{g}_k(x_2)]$.
\end{theorem*}
\begin{proof}
The proof of this theorem starts with a similar strategy as in \cite{spde}.
Without the loss of generality, we assume the variance parameter of the sGP $\sigma = 1$ and $\Omega = [0,1]$.
We denote the 3rd order Sobolev space as $H^3(\Omega) = \{f \in \mathcal{L}^2(\Omega): D^q f \in \mathcal{L}^2(\Omega) \ \forall |q| \leq 3\}$. We then define the constrained Sobolev space $\Sobolev$ as:
$$\Sobolev = \{f \in H^3(\Omega): f(0) = f'(0) = 0\}.$$
Since $Lf = 0$ implies $f \in \text{span}\{\cos(\alpha x), \sin(\alpha x)\}$, it is clear that
\begin{equation}
    \begin{aligned}
        \ip[\Sobolev]{f}{h} := \ip[\Omega]{Lf}{Lh}
                            = \int_\Omega Lf(x) Lh(x) dx
    \end{aligned}
\end{equation}
defines an inner product for $f, h \in \Sobolev$.

Let $b > 1$ and $\Tilde{\Omega} = [0,b) \supset \Omega$, then for each $f \in H^3(\Omega)$, there exists a zero-extension $f_0 \in H^3(\Tilde{\Omega})$ such that $f_0(x) = f(x)$ for $x \in \Omega$, and $f_0(b) = f'_0(b) = 0$.
Hence, we assume without the loss of generality that for each $f \in \Sobolev$, $f(1) = f'(1) =0$ from now. This implies that the differential operator $D^q$ in $\Sobolev$ has adjoint operator $(D^q)^* = (-1)^q D^q$ for each $q \in \mathbb{Z}$.

Define $\Sobolev_k = \text{span} \{\mathbb{B}_k\}$.
Note $\Sobolev_k \subset \Sobolev$ by our construction of the B-spline basis.
This implies if $f(x) \in \Sobolev$, then there exists a projection $\Tilde{f}(x) = \sum_{i=1}^k w_i \varphi_i(x) \in \Sobolev_k$ which satisfies:
\begin{equation}
    \begin{aligned}
        \ip[\Sobolev]{f-\tilde{f}}{\tilde{h}} = \ip[\Sobolev]{f}{\tilde{h}} - \ip[\Sobolev]{\tilde{f}}{\tilde{h}} = 0, \  \forall \tilde{h} \in \Sobolev_k.
    \end{aligned}
\end{equation}

\noindent To prove $\lim_{k\to \infty}\C_{k}(x_1,x_2) = \C(x_1,x_2)$, note that for any $f,h \in \Sobolev$,
\begin{equation}
    \begin{aligned}
        \Cov[\ip[\Sobolev]{g}{f},\ip[\Sobolev]{g}{h}] &:=  \Cov[\ip[\Omega]{Lg}{Lf},\ip[\Omega]{Lg}{Lh}] \\
                    &= \Cov[\ip[\Omega]{ \xi}{Lf},\ip[\Omega]{ \xi}{Lh}] \\
                    &=  \ip[\Omega]{Lf}{Lh}, \\
                    &=  \ip[\Sobolev]{f}{h}.
    \end{aligned}
\end{equation}
The second equality follows from the definition of the sGP, and the third equality follows from the property of Gaussian white noise \citep{harvey1990forecasting}.

\noindent Since the B-spline approximation $\tilde{g}_k \in \Sobolev_k$, we know $$\ip[\Sobolev]{\tilde{g}_k}{f}= \ip[\Sobolev]{\tilde{g}_k}{f-\tilde{f} + \tilde{f}} = \ip[\Sobolev]{\tilde{g}_k}{\tilde{f}} + \ip[\Sobolev]{\tilde{g}_k}{f - \tilde{f}} = \ip[\Sobolev]{\tilde{g}_k}{\tilde{f}}.$$
Using this result and the fact that the B-spline approximation $\tilde{g}_k$ is a least square solution, we have
\begin{equation}
    \begin{aligned}
        \Cov\bigg[\ip[\Sobolev]{\tilde{g}_k}{f},\ip[\Sobolev]{\tilde{g}_k}{h}\bigg] &= \Cov\bigg[\ip[\Sobolev]{\tilde{g}_k}{\tilde{f}},\ip[\Sobolev]{\tilde{g}_k}{\tilde{h}}\bigg] \\
        &= \Cov\bigg[\ip[\Omega]{\xi}{L\tilde{f}},\ip[\Omega]{\xi}{L\tilde{h}}\bigg] \\
        &= \ip[\Sobolev]{\tilde{f}}{\tilde{h}}.
    \end{aligned}
\end{equation}

\noindent Let $\C_s(x) = \C(s,x)$ denote the covariance function of the sGP defined at any $s \in \Omega$. Based on the previous result in \cref{theorem:sGP_Property}, we know $\C_s(x) \in \Sobolev$ and $L \C_s(x) = \frac{1}{\alpha} \sin[\alpha(s-x)^{+}]$ is the Green function of $L$. Its projection into $\Sobolev_k$ is denoted as $\tilde{\C}_s(x)$

\begin{lemma}\label{lemma:RK}
Given the same setting in the main theorem
\begin{equation}
    \begin{aligned}
    \C(x_1,x_2) &= \Cov \big[\ip[\Sobolev]{g}{{\C}_{x_1}},\ip[\Sobolev]{g}{{\C}_{x_2}}  \big] \\
        \C_{k}(x_1,x_2) &= \Cov \bigg[\ip[\Sobolev]{\tilde{g}_k}{\tilde{\C}_{x_1}},\ip[\Sobolev]{\tilde{g}_k}{\tilde{\C}_{x_2}}  \bigg].
    \end{aligned}
\end{equation}
\end{lemma} 

\begin{proof}

The first part directly follows from the proof in \cref{theorem:sGP_Property}. The second part can be proved using the fact that $L$ is self-adjoint and $L \C_{x_1}(x) = \frac{1}{a} \sin[a(x_1-x)^{+}]$ is the Green function, which implies $\ip[\Sobolev]{\varphi_i}{\tilde{\C}_{x_1}} = \varphi_i(x_1)$ for each $i \in [k]$. The detailed proof proceeds as the following.

\noindent By construction of the B-spline approximation, 
\begin{equation}
    \begin{aligned}
    \C_{k}(x_1,x_2) &= \text{Cov}\bigg[ \sum_i w_i \varphi_i(x_1), \sum_i w_i \varphi_i(x_2)  \ \bigg] \\
    &= \text{Cov}\bigg[ \boldsymbol{\Phi}(x_1) ^T \boldsymbol{w}, \boldsymbol{\Phi}(x_2) ^T \boldsymbol{w} \bigg] \\
    &= \boldsymbol{\Phi}(x_1) ^T \Sigma_{\boldsymbol{w}} \boldsymbol{\Phi}(x_2),
    \end{aligned}
\end{equation}
where $\boldsymbol{\Phi}(x) = [\varphi_1(x), ..., \varphi_k(x)]^T$. 

\noindent Since $\tilde{\C}_{x_1}(x)$ is the projection of ${\C}_{x_1}(x)$ to $\Sobolev_k$, $\tilde{\C}_{x_1}(x) = \sum_i w_{x_1,i} \varphi_i(x)$ for some weights $\boldsymbol{w}_{x_1} = [w_{x_1,1}, ... , w_{x_1,k}]^T$. The same argument can be used for $\tilde{\C}_{x_2}$.
Therefore
\begin{equation}
    \begin{aligned}
        \text{Cov}\bigg[\ip[\Sobolev]{\tilde{g}_k}{\tilde{\C}_{x_1}},\ip[\Sobolev]{\tilde{g}_k}{\tilde{\C}_{x_2}}  \bigg] &= \ip[\Sobolev]{\tilde{\C}_{x_1}}{\tilde{\C}_{x_2}} \\
        &= \ip[\Sobolev]{\boldsymbol{\Phi}(x) ^T \boldsymbol{w}_{x_1}}{\boldsymbol{\Phi}(x) ^T \boldsymbol{w}_{x_2}} \\
        &= \ip[\Omega]{L\boldsymbol{\Phi}(x) ^T \boldsymbol{w}_{x_1}}{L\boldsymbol{\Phi}(x) ^T \boldsymbol{w}_{x_2}} \\
        &= \boldsymbol{w}_{x_1}^T \Sigma^{-1}_{\boldsymbol{w}} \boldsymbol{w}_{x_2}.
    \end{aligned}
\end{equation}
Solving $\ip[\Sobolev]{\tilde{\C}_{x_1}}{\varphi_i} = \ip[\Sobolev]{\C_{x_1}}{\varphi_i}$ for each $i \in [k]$ yields that $\boldsymbol{w}_{x_1} = \Sigma_{\boldsymbol{w}}\boldsymbol{\omega}_{x_1}$, where $\boldsymbol{\omega}_{x_1} \in \mathbb{R}^k$ with $i$th element ${\omega}_{x_1,i} = \ip[\Sobolev]{\C_{x_1}}{\varphi_i}$. 
Hence it only remains to show $\boldsymbol{\omega}_{x_1} = \boldsymbol{\Phi}(x_1)$, which holds because
\begin{equation}
\begin{aligned}
 \omega_{x_1, i} &= \ip[\Sobolev]{\C_{x_1}}{\varphi_i} \\
 &= \ip[\Omega]{L\C_{x_1}}{L \varphi_i} \\
 &= \ip[\Omega]{L^{*}L\C_{x_1}}{ \varphi_i} \\
 &= \varphi_i(x_1), \ \forall i \in [k].
\end{aligned}
\end{equation}
The last equality holds because $L$ is a self-adjoint operator and $L\C_{x_1}$ is the Green function. This lemma is hence proved.
\end{proof}

\noindent Using the above result and \cref{lemma:RK}, it suffices to prove that 
\begin{equation}\label{equ:InnerConv}
    \begin{aligned}
        \ip[\Sobolev]{\tilde{\C}_{x_1}}{\tilde{\C}_{x_2}} \rightarrow \ip[\Sobolev]{\C_{x_1}}{\C_{x_2}},
    \end{aligned}
\end{equation}
as $k \rightarrow \infty$. 
For this step, we will use the following lemma on the spline approximation:

\begin{lemma}\label{lemma:SplineApproxi}
Given the same setting in the main theorem, define the norm $||f||_\Sobolev = \ip[\Sobolev]{f}{f}^{1/2}$ for $f \in \Sobolev$ then 
\begin{equation}
    \begin{aligned}
    ||\C_{s} -\tilde{\C}_{s}||_{\Sobolev}  &= O(1/k),
    \end{aligned}
\end{equation}
for each $s \in \Omega$.
\end{lemma} 
\begin{proof}
The proof of this lemma mostly follows from the result in \cite{schultz1969approximation}. 
Given $\Sobolev_k$ is a spline space with degree $3$ and mesh size $1/k$ and $\C_s \in H^3(\Omega)$, by theorem 3.3 in \cite{schultz1969approximation} we have
\begin{equation}
    \begin{aligned}
        ||D^q(\C_s - \tilde{\C}_s)||_{H^2(\Omega)} \leq c_q\bigg(\frac{1}{k}\bigg)^{3-q}
    \end{aligned}
\end{equation}
for each $0\leq q \leq 2$, where $c_q$ is a constant that depends on $||\C_s||_{H^3(\Omega)}$ but does not depend on $k$.
Note that 
\begin{equation}
    \begin{aligned}
        ||\C_s - \tilde{\C}_s||^2_{\Sobolev} &= \alpha^4||\C_s - \tilde{\C}_s||^2_{\mathcal{L}^2} + ||D^2(\C_s - \tilde{\C}_s)||^2_{\mathcal{L}^2} - 2\alpha \ip[\Omega]{(\C_s - \tilde{\C}_s)}{D^2(\C_s - \tilde{\C}_s)} \\
        &= \alpha^4||\C_s - \tilde{\C}_s||^2_{\mathcal{L}^2} + ||D^2(\C_s - \tilde{\C}_s)||^2_{\mathcal{L}^2} +  2\alpha||D(\C_s - \tilde{\C}_s)||^2_{\mathcal{L}^2},
    \end{aligned}
\end{equation}
where the second equality holds since $D^* = -D$. The lemma is proved since $||.||_\Sobolev$ and $||.||_{H^2(\Omega)}$ are equivalent norm.

\end{proof}

\noindent Using \cref{lemma:SplineApproxi}, \cref{equ:InnerConv} can be proved with Cauchy Schwarz inequality and the fact that the sequence $||\tilde{\C}_{x_2}||_\Sobolev$ is bounded,
\begin{equation}
    \begin{aligned}
        \bigg|\ip[\Sobolev]{\tilde{\C}_{x_1}}{\tilde{\C}_{x_2}} - \ip[\Sobolev]{\C_{x_1}}{\C_{x_2}}\bigg| &= 
        \bigg|\ip[\Sobolev]{\tilde{\C}_{x_1} - \C_{x_1}}{\tilde{\C}_{x_2}} - \ip[\Sobolev]{\C_{x_1}}{\C_{x_2} - \tilde{\C}_{x_2}}\bigg| \\
        &\leq \big|\big| \tilde{\C}_{x_1} - \C_{x_1} \big|\big|_\Sobolev \big|\big| \tilde{\C}_{x_2} \big|\big|_\Sobolev + \big|\big| \tilde{\C}_{x_2} - \C_{x_2} \big|\big|_\Sobolev \big|\big| \C_{x_1} \big|\big|_\Sobolev \\
        & = c/k,
    \end{aligned}
\end{equation}
where $c$ is some constant independent of $k$. The theorem is hence proved.
\end{proof}

\end{document}